\documentclass[sigconf,edbt]{acmart-edbt2019}

\usepackage{booktabs} 

\usepackage{balance}  

\setcopyright{rightsretained}

\acmDOI{}

\acmISBN{XXX-X-XXXXX-XXX-X}

\usepackage{xcolor}
\usepackage{amsmath,amsfonts,amsthm}
\newcommand{\urlNewWindow}[1]{\href[pdfnewwindow=true]{#1}{\nolinkurl{#1}}}
\usepackage{mathtools}
\usepackage{hhline}
\usepackage{stmaryrd}
\usepackage{mathpartir}
\usepackage{geometry}
\usepackage{upgreek}
\usepackage{tikz}
\usepackage{booktabs} 
\usepackage{algorithm}
\usepackage[noend]{algpseudocode}
\algnewcommand{\LeftComment}[1]{\Statex \(\triangleright\) #1}

\algrenewcommand\algorithmicindent{0.5em}%

\usepackage{hyperref}
\usepackage{multirow}

\usepackage{listings}
\makeatletter
 \def\hlinewd#1{%
   \noalign{\ifnum0=`}\fi\hrule \@height #1 \futurelet
    \reserved@a\@xhline}
 \makeatother

 \makeatletter
 \def\algbackskip{\hskip-\ALG@thistlm}
 \makeatother

\usepackage{subfig}
\usepackage{graphicx}

 \usetikzlibrary{arrows,shapes,snakes,automata,backgrounds,petri}
 \usepackage{pgfplots}
 \pgfplotsset{compat=1.12}

\usetikzlibrary{positioning}

\usepackage{subfig}
\usepackage[utf8]{inputenc} 
\usepackage{xspace}
\usepackage{enumitem}
\usepackage{xcolor}

\begin{document}
\title{Approximate Evaluation of Label-Constrained \\ Reachability Queries}
%

\author{Stefania Dumbrava}
 \affiliation{%
   \institution{ENS Rennes \& CNRS IRISA \& INRIA}
 }
 \email{stefania.dumbrava@inria.fr}
 
\author{Angela Bonifati}
 \affiliation{%
   \institution{University of Lyon 1 \& CNRS LIRIS}
 }
 \email{angela.bonifati@univ-lyon1.fr}
 
\author{Amaia Nazabal Ruiz Diaz}
 \affiliation{%
   \institution{University of Lyon 1 \& CNRS LIRIS}
 }
 \email{amaia.nazabal-ruiz-diaz@univ-lyon1.fr}
 
\author{Romain Vuillemot}
 \affiliation{%
   \institution{Ecole Centrale Lyon \& CNRS LIRIS}
 }
 \email{romain.vuillemot@ec-lyon.fr}

\renewcommand{\shortauthors}{}

\begin{abstract}
The current surge of interest in graph-based data models mirrors the usage of increasingly complex reachability queries, 
as witnessed by recent analytical studies on real-world graph query logs. Despite the maturity of graph DBMS capabilities, 
complex label-constrained reachability queries, along with their corresponding aggregate versions, remain difficult to evaluate. 
In this paper, we focus on the approximate evaluation of counting label-constrained reachability queries. 
We offer a human-explainable solution to graph Approximate Query Processing (AQP). This consists of a summarization algorithm (GRASP),
as well as of a custom visualization plug-in, which allows users to explore the obtained summaries.
We prove that the problem of node group minimization, associated to the creation of GRASP summaries, is NP-complete. Nonetheless, our GRASP summaries 
are reasonably small in practice, even for large graph instances, and guarantee approximate graph query answering, paired with controllable error estimates.
We experimentally gauge the scalability and efficiency of our GRASP algorithm, and verify the accuracy and error estimation of the graph AQP module. 
To the best of our knowledge, ours is the first system capable of handling visualization-driven approximate graph analytics for 
complex label-constrained reachability queries.
\end{abstract}

\maketitle

\section{Introduction}\label{sec:intro}
A tremendous amount of information stored in graph-shaped format can be
inspected, by leveraging the already mature query capabilities of graph DBMSs
\cite{AnglesABBFGLPPS18, AnglesABHRV17, 2018Bonifati}.
However, arbitrarily complex graph reachability queries, entailing rather intricate and possibly recursive
graph patterns (required to extract friendship relationships, in social
networks, or protein-to-protein interactions, in biological
networks), prove difficult to evaluate, even on small-sized graph
datasets \cite{BaganBCFLA17, KhandelwalYY0S17, SzarnyasPAMPKEB18}. On the other hand, the usage of these queries has radically increased in
real-world graph query logs, as shown by recent empirical studies on Wikidata and DBPedia corpuses \cite{BonifatiMT17, Malyshev18}.
For example, the percentage of SPARQL property paths has grown from $15\%$ to $40\%$, for organic Wikidata queries, 
from 2017 to beginning 2018 \cite{Malyshev18}.

In this paper, we 
aim to offer \emph{human-explainable
approximate graph query processing}, by focusing on regular path
queries that identify graph paths, through regular expressions on edge labels.
Interactive data analytics, especially addressing the relational case,
has provided query result estimates with bounded
errors, relying on approximate query processing (AQP) techniques \cite{ChaudhuriDK17, PengZWP18}.
While many AQP systems have been proposed for diverse sets of SQL
aggregate queries, exploiting biased sampling (e.g., BlinkDB
\cite{AgarwalMPMMS13}) and online aggregation for the
mass of value distributions (e.g., G-OLA \cite{ZengADAS15}) or for rare populations 
(e.g., IDEA \cite{GalakatosCZBK17}),
AQP for graph databases has remained unexplored so far.

We target VAGQP (Visualization-driven Approximate Graph
Query Processing) for counting 
reachability queries that are label-constrained.
The exact evaluation of these queries is $\#P-$complete, following a result on the enumeration of
simple graph paths, due to Valiant \cite{Valiant79}. We rather focus on its \emph{approximation}
and show how it can be effectively used in practice. 

Sampling approaches, customarily used for relational AQP,
are not directly applicable to graph processing, due to the lack of the
linearity assumption \cite{IyerPVCASS18} \footnote{The linear
relationship between the sample size and execution time typical of relational query processing falls apart
in graph query processing.}. In view of this, we rely on a novel, query and workload-driven, graph
summarization technique that provides the baseline data structure for VAGQP.

While designing VAGQP, we devoted particular attention to facilitating the exploratory
steps carried out during the analytical process. We thus conceived an
exploratory visualizer, capable of showing the degree of approximation of
the graph summaries, for human understanding and result explanation.
Consequently, we offer alternative and more compact
visualizations of the envisioned graph summaries, using \emph{linked treemaps}.

Our techniques are seamlessly designed for \emph{property graphs}, attaching data
values to property lists on both nodes and edges and, thus, keeping pace with the
latest developments in graph databases and graph query languages \cite{Angles18, AnglesABBFGLPPS18}.

\emph{To the best of our knowledge, ours is the first work on approximate
graph analytics that addresses counting estimation on top of navigational graph queries}. 
Furthermore, we tackle both \emph{graph visualization} and \emph{programming APIs for graph query languages},
identified among the top 3 graph processing challenges, in a recent user survey \cite{SahuMSLO17}.
We first provide efficient in-DBMS translations of queries formulated on an initial graph to
corresponding \emph{approximate queries}, on top of the graph's GRASP summary. We then extensively rely
on exploratory data analysis to help users identify, in the linked treemap summary encoding, the regions
that are relevant to the approximate query evaluation process. 

We illustrate our problem through the running example below.
\begin{example}[Graph AQP for Social Network Advertising]\label{ref:ads}
Let $\mathcal{G}_{SN}$ (see Fig. \ref{ref:sn}) represent a social network, whose schema is inspired by the LDBC benchmark.
Entities are people (type \texttt{Person}, $P_i$) that \emph{know} ($l_0$) and/or \emph{follow} ({\color{blue}$l_1$})
either each other or certain forums (type \texttt{Forum}, $F_i$). These are \emph{moderated} ({\color{green}$l_2$}) by specific persons and
can \emph{contain} ({\color{purple}$l_3$}) messages/ads (type \texttt{Message}, $M_i$), to which persons can \emph{author} ({\color{orange}$l_4$})
other messages in \emph{reply} ({\color{red}$l_5$}).
$\mathcal{G}_{SN}$ exemplifies a graph instance adhering to the property graph model (PGM) that we will define in Section \ref{sec:prelim}.
Our goal is to \textbf{perform graph AQP to obtain high-accuracy, fast, query estimates}.
A practical application in this scenario would be leveraging AQP to obtain targeted advertising markers in social networks.
To make use of the heterogeneity of real-world networks, we need to express aggregate queries in a query language
allowing labeled constraints. This corresponds to a dialect of Regular Path Queries \cite{Wood2012,AnglesABHRV17,Angles18} and
suffices to express the aggregate RPQ-based query types illustrated below \footnote{For ease of exposition, their translation in a high-level syntax is reported in Figure \ref{ref:qrpq} in Section 2.}.

\noindent {\textbf{Simple and Optional Label.}} The count of node pairs that satisfy $Q_1$, i.e., $() {\color{red}\xrightarrow{l_5}} ()$, captures the number of ad \emph{reactions},
 while the corresponding count for $Q_2$, i.e., $() {\color{green}\xrightarrow{l_2?}} ()$ indicates the number of \emph{actual and potential moderators}.

\noindent {\textbf{Kleene Plus/Kleene Star.}} The number of the \emph{connected acquaintances}/\emph{potentially connected acquaintances} is the count of node pairs satisfying
 $Q_3$, i.e., $() \leftarrow{l^{+}_0} ()$, respectively, $Q_4$, i.e., $() \leftarrow{l^{*}_0} ()$

\noindent {\textbf{Disjunction.}}
The number of the \emph{targeted subscribers} is the sum of counting all
 node pairs satisfying $Q_5$, i.e., $() {\color{orange}\xleftarrow{l_4}} ()$ or $() {\color{blue}\xleftarrow{l_1}} ()$.

\noindent {\textbf{Conjunction.}}
 The \emph{direct reach} of a company via its page ads is the count of all node pairs satisfying $Q_6$, i.e., $() {\color{orange}\xleftarrow{l_4}} () {\color{red}\xrightarrow{l_5}} ()$.

\noindent {\textbf{Conjunction with Property Filters.}}
 Recommendation systems can then further refine the $Q_6$ estimates by also taking into account particular
 properties associated to nodes. This can be done by determining the \emph{direct demographic reach} targeting people within a certain age group, e.g., 18-24, 
 counting all node pairs that satisfy $Q_7$, i.e. $(x) {\color{orange}\xleftarrow{l_4}} () {\color{red}\xrightarrow{l_5}} ()$,
 where $x.age \geq 18$ and $x.age \leq 24$.
\begin{figure}[t!]
\centering
  \includegraphics[width=\columnwidth]{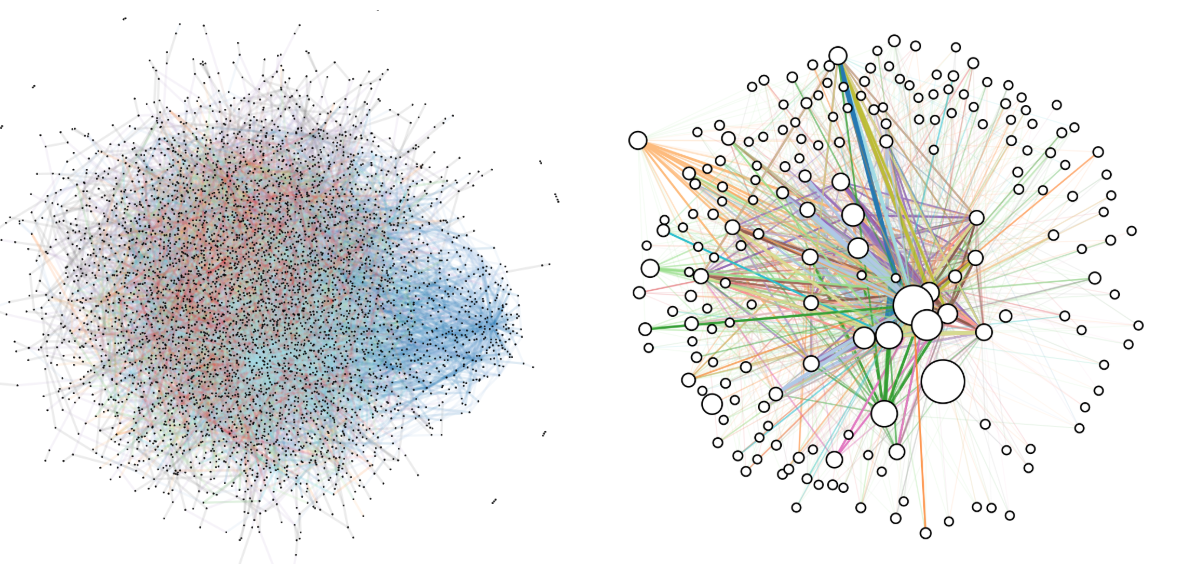}
  \caption{Original Graph (left) and Graph Summary (right) for 5K LDBC social network using
node link visualization.} 
    \label{fig:node-link-treemap}
\end{figure}

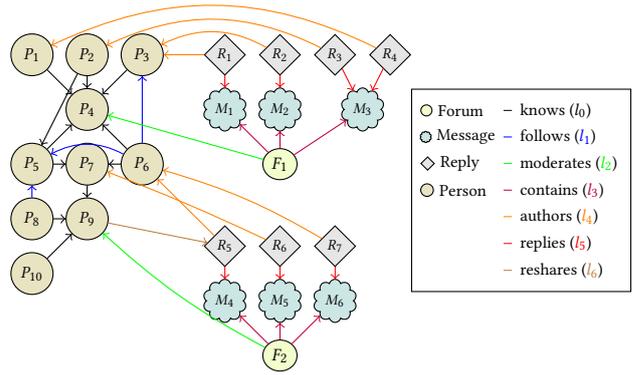
\begin{figure}[t!]
   \begin{tikzpicture}[scale=1.45, every node/.style={scale=0.7}]
  \tikzstyle{fe}=[ellipse, draw, thin,fill=lime!20, scale=0.7]
  \tikzstyle{mr}=[cloud, draw, thin,fill=teal!20, scale=0.6]
  \tikzstyle{rd}=[diamond, draw, thin,fill=gray!20, scale=0.6]
  \tikzstyle{mm}=[rectangle, maximum height=0cm, draw, thin]

  \tikzstyle{every state}=[fill=olive!20,draw=black,text=black, inner sep=0.8pt]
    \node[state]         (P1)  at (0,2.5)    {$P_1$};
    \node[state]         (P2)  at (0.5,2.5)  {$P_2$};
    \node[state]         (P3)  at (1,2.5)    {$P_3$};
    \node[state]         (P4)  at (0.5,2)    {$P_4$};
    \node[state]         (P5)  at (0,1.5)    {$P_5$};
    \node[state]         (P6)  at (1,1.5)    {$P_6$};
    \node[state]         (P7)  at (0.5,1.5)  {$P_7$};
    \node[state]         (P8)  at (0,1)      {$P_8$};
    \node[state]         (P9)  at (0.5,1)    {$P_9$};
    \node[state]         (P10) at (0,0.5)    {$P_{10}$};
    \node[rd]         (R1)  at (1.75,2.5)    {\scalebox{1.5}{$R_1$}};
    \node[rd]         (R2)  at (2.25,2.5)    {\scalebox{1.5}{$R_2$}};
    \node[rd]         (R3)  at (2.75,2.5)    {\scalebox{1.5}{$R_3$}};
    \node[rd]         (R4)  at (3.25,2.5)    {\scalebox{1.5}{$R_4$}};
    \node[mr]         (M1)  at (1.75,2)      {\scalebox{1.5}{$M_1$}};
    \node[mr]         (M2)  at (2.25,2)      {\scalebox{1.5}{$M_2$}};
    \node[mr]         (M3)  at (3,2)         {\scalebox{1.5}{$M_3$}};
    \node[fe]            (F1)  at (2.25,1.5)    {\scalebox{1.5}{$F_1$}};
    \node[fe]            (F2)  at (2.25,-0.25)  {\scalebox{1.5}{$F_2$}};
    \node[mr]         (M4)  at (1.75,0.25)   {\scalebox{1.5}{$M_4$}};
    \node[mr]         (M5)  at (2.25,0.25)   {\scalebox{1.5}{$M_5$}};
    \node[mr]         (M6)  at (2.75,0.25)   {\scalebox{1.5}{$M_6$}};
    \node[rd]         (R5)  at (1.75,0.75)   {\scalebox{1.5}{$R_5$}};
    \node[rd]         (R6)  at (2.25,0.75)   {\scalebox{1.5}{$R_6$}};
    \node[rd]         (R7)  at (2.75,0.75)   {\scalebox{1.5}{$R_7$}};
   \path (P1) edge[->]   node {}  (P4)
	 (P2) edge[->]   node {}  (P4)
	 (P3) edge[->]   node {}  (P4)
	 (P5) edge[->]   node {}  (P4)
	 (P6) edge[->]   node {}  (P4)
	 (P6) edge[->, color=blue]   node {}  (P3)
	 (P6) edge[->, bend right=28, above, color=blue]   node {}  (P5)
	 (P2) edge[->]   node {}  (P5)
	 (P5) edge[->]   node {}  (P7)
	 (P6) edge[->]   node {}  (P7)
	 (P7) edge[->]   node {}  (P9)
	 (P8) edge[->]   node {}  (P9)
	 (P8) edge[->, color=blue]   node {}  (P5)
	 (P10) edge[->]  node {}  (P9)
	 (R1) edge[->, color=red]   node {}  (M1)
	 (R2) edge[->, color=red]   node {}  (M2)
	 (R3) edge[->, color=red]   node {}  (M3)
	 (R4) edge[->, color=red]   node {}  (M3)
	 (F1) edge[->, color=purple]   node {}  (M1)
	 (F1) edge[->, color=purple]   node {}  (M2)
	 (F1) edge[->, color=purple]   node {}  (M3)
	 (F2) edge[->, color=purple]   node {}  (M4)
	 (F2) edge[->, color=purple]   node {}  (M5)
	 (F2) edge[->, color=purple]   node {}  (M6)
	 (R5) edge[->, color=red]   node {}  (M4)
	 (R6) edge[->, color=red]   node {}  (M5)
   	 (R7) edge[->, color=red]   node {}  (M6)
   	 (F1) edge[->, color=green]   node {}  (P4)
   	 (F2) edge[->, bend left=7, above, color=green]   node {}  (P9)
   	 (R1) edge[->, color=orange]   node {}  (P3)
   	 (R5) edge[->, color=orange]   node {}  (P6)
   	 (R6) edge[->, bend right=2, above, color=orange]   node {}  (P7)
   	 (R7) edge[->, bend right=7, above, color=orange]   node {}  (P6)
   	 (R2) edge[->, bend right=28, above, color=orange]   node {}  (P3)
   	 (R3) edge[->, bend right=26, above, color=orange]   node {}  (P2)
   	 (R4) edge[->, bend right=26, above, color=orange]   node {}  (P1)
   	 (P9) edge[->, color=brown]   node {}  (R5)
   ;
    \matrix(table) [draw,right, yshift=-5pt] at (current bounding box.east) {
      \node [fe,label=right:Forum] {};
      & \node [state,opacity=0,text opacity=1, label=right:knows ($l_0$),scale=0.3] {\par\noindent\rule{15pt}{1pt}}; \\
      \node [mr,label=right:Message] {};
      & \node [state,opacity=0,text opacity=1, label=right:follows ({\color{blue}$l_1$}),scale=0.3] {{\color{blue}\par\noindent\rule{15pt}{1pt}}};  \\
      \node [rd,label=right:Reply] {}; & \node [state,opacity=0,text opacity=1, label=right:moderates ({\color{green}$l_2$}),scale=0.3] {{\color{green}\par\noindent\rule{15pt}{1pt}}}; \\
      \node [state,label=right:Person,scale=0.3] {};
        & \node [state,opacity=0,text opacity=1, label=right:contains ({\color{purple}$l_3$}),scale=0.3] {{\color{purple}\par\noindent\rule{15pt}{1pt}}}; \\
        & \node [state,opacity=0,text opacity=1, label=right:authors ({\color{orange}$l_4$}),scale=0.3]  {{\color{orange}\par\noindent\rule{15pt}{1pt}}}; \\
        & \node [state,opacity=0,text opacity=1, label=right:replies ({\color{red}$l_5$}),scale=0.3]  {{\color{red}\par\noindent\rule{15pt}{1pt}}}; \\
        & \node [state,opacity=0,text opacity=1, label=right:reshares ({\color{brown}$l_6$}),scale=0.3]  {{\color{brown}\par\noindent\rule{15pt}{1pt}}}; \\
      };
     \end{tikzpicture}
     \vspace{0.5mm}
  \caption{Example Social Graph $\mathcal{G}_{SN}$}
  \label{ref:sn}
\end{figure}
\end{example}

\noindent {\bf Contributions.}
The paper tackles the following key challenges:

\noindent \underline{Query-oriented Graph Summarization}. The direct evaluation
of the above analytical queries on the original graph can be costly in terms
of runtime. To address this, we design the GRASP summarization algorithm, aimed
at preserving label-constrained reachability information and at book-keeping 
additional AQP-relevant statistics in the summary's node properties.
Unlike existing graph summarization methods for labeled graphs,
based on grouping, compression or influence \cite{LiuDSK16}, 
GRASP explicitly inspects the query workload and takes into account the labels
deemed as important, in order to provide a compression suitable for
graph analytical applications.
The produced graph summary, grouping nodes into supernodes (SN) and merging them in hypernodes (HN), is
guaranteed to be encoded as a property graph, similarly to the original
graph, thus allowing the evaluation of approximate queries directly in the graph
database. We additionally provide an alternative visual representation of the obtained graph summary, which
enables the user to understand its compactness and compression ratio (see Figure \ref{fig:node-link-treemap} right-hand-side 
for a grasp of the latter, on a 5K LDBC social network graph \footnote{More details about Figure 1
will be given in Section 3.3.}.) 

\noindent \underline{Blending Visual and Numerical Graph AQP}.
The node-link representation of the graph summary, as illustrated in Figure \ref{fig:node-link-treemap}, is not tailored for
VAGQP. In view of this, we introduce linked treemaps, a visualization overlay,
considering supernodes and hypernodes, along with their interconnecting
edges in the graph summary. Inspired by previous work on semantic
substrates for network exploration \cite{fekete03,ArisS07}, we rely on treemaps
as semantic overlays for multi-labeled graphs and
adapt them to our purposes.
We let the user locate
the densest/sparsest graph regions to be inspected by the VAGQP engine
on the linked treemaps. 
VAGQP is coupled with an automatic translation of the
queries, expressed in an RPQ dialect, into corresponding ones, in the same
dialect, on the graph summary, further leveraging node properties. 

\noindent \underline{Small Error Bounds and Experimental Analysis}.
First, we define a seamless translation, guaranteeing that queries on the summary and on the initial graph are expressible 
in the same fragment. Based on this, we experimentally illustrate the small relative errors of GRASP-based AQP, for each supported query type.

Our analysis also proves the effectiveness and efficiency of the GRASP algorithm on datasets with varying 
degrees of heterogeneity and sizes. We show that the error bounds for various query workloads, using disjunctions, single-label, Kleene-star, Kleene-plus, 
Optionality and Concatenations of different lengths (up to the sizes observed in practical studies), are relatively small and bearable. 
Apart from the high accuracy, we measure the relative response time of answering counting label-constrained reachability queries
on summaries, compared to that on the original graph. 
For a small query subset supported by both systems, we
compare with SumRDF \cite{Stefanoni:2018:ECC}, a baseline RDF graph summary, and show the better accuracy and
query runtime of our approach. 
Finally, we show the utility of linked treemaps in restricting approximate query evaluation to specific 
graph regions.

The paper is organized as follows. In Section \ref{sec:prelim}, we present the considered graph model and query language.
In Section \ref{sec:gsum}, we detail our proposed query-oriented, parametric, graph summarization algorithm (GRASP) along with 
its visual, linked treemap, counterpart. We explain the query translation needed for performing AQP on the summarized graph, in Section \ref{sec:qtrans}, and the
visualization techniques leveraged to identify summarization parameters, in Section \ref{sec:vaqp}. To gauge the performance of AQP on
GRASP-summaries over graphs with various characteristics, we conduct an
extensive experimental assessment in Section \ref{sec:exp}. We present related work in Section \ref{sec:relwork} and conclude in Section \ref{sec:concl}.

\section{Preliminaries}\label{sec:prelim}
\noindent {\bf Graph Model.} We take the \emph{property graph model} (PGM)~\cite{Angles18} as our underlying formalism (see Fig \ref{fig:pgm}).
Graph instances are thus multi-edge digraphs; its objects are represented by typed, data vertices and the relationships between these, as typed, 
labeled edges. Both vertices and edges can have any number of associated \emph{properties} (key/value pairs);
we denote the set of all properties as $\mathcal{P}$.
Let $\upgamma^{}_{V}$ and $\upgamma^{}_{E}$ be disjoint vertex and edge types, whose elements we denote as $L_V$, resp., $L_E$.
We call $\mathcal{G} = (V,E)$, where $E \subseteq V \times L_E \times V$, a \emph{graph instance}, with $V$ and $E$, disjoint 
sets of vertices $v$ and, respectively, of edge labels $l_e$. Vertices $v$ have a label id, $l_v$, of type $L_V$, and a set of 
property labels (attributes $l_i$), each with a certain, potentially undefined, respective term value $t_1, \ldots, t_n$. 
We denote the vertex and edge labeling function as $\gamma : E \cup V \rightarrow L_E \cup L_V$ and the 
schema function, as $\sigma : L_E \cup L_V \rightarrow \mathcal{P}$. 
We henceforth use a binary notation for edges 
and, given $e \in E$, $e = l(v_1,v_2)$, we abbreviate $l$ as $\gamma(e)$, $v_1$ as $e.1$ and $v_2$ as $e.2$. For a given edge label, $l$, we abbreviate 
its number of occurrences in $\mathcal{G}$ as $\#l$. For a label set, $\Lambda = \{l^{1}_e, \ldots, l^{n}_e\}$, its associated 
\emph{frequency list} is $\vec{\Lambda} = [l_1, \ldots, l_n]$, where $l_1 = l^{i_1}_e, \ldots, l_n = l^{i_n}_e$ and
$\{i_1, \ldots, i_n\}$ is a permutation of $\{1,\ldots,n\}$, such that $\#l^{i_1}_e \geq \ldots \geq \#l^{i_n}_e$.
For a graph $G = (V, E)$, its set of edge labels is $\Lambda(\mathcal{G}) = \{\gamma(e) ~|~ e \in E\}$ and,
for a $\mathcal{G}$-subgraph, $\mathcal{G}'= (V', E')$, its set of incoming/outgoing edge labels is
$\Lambda_{+}(\mathcal{G'}) = \{\gamma(e) ~|~ e \in E \wedge e.2 \in V' \wedge e.1 \notin V'\}$ and 
$\Lambda_{-}(\mathcal{G'}) = \{\gamma(e) ~|~ e \in E \wedge e.1 \in V' \wedge e.2 \notin V'\}$.

\begin{figure}[t!]
{\small 
\[
\arraycolsep=1.2pt\def\arraystretch{1}
\begin{array}{llll}
\hline
\text{Edges} ~ &  e & ~{\color{red}::=}~ & \epsilon ~|~ \lambda l_e. \; (p_1, t_1), \ldots (p_n, t_n), \text{for } l_e \in L_E \cup \mathcal{V},\\ 
&&& \gamma(e) = l_{e}, \text{ and } \sigma(l_{e}) = \{p_1, \ldots, p_n\} \subseteq \mathcal{P} \\ \hline
\text{Vertices} ~ & v & ~{\color{red}::=}~ & x \in \mathcal{V} ~|~ \lambda l_{v}. \; (p_1, t_1), \ldots (p_n, t_n), \text{for } l_v \in L_V \cup \mathcal{V}, \\ 
&&& \gamma(v) = l_{v}, \text{ and } \sigma(l_{v}) = \{p_1, \ldots, p_n\} \subseteq \mathcal{P} \\ \hline
\text{Terms}    ~ & t & ~{\color{red}::=}~ & \bot ~~|~~ c \in \mathcal{C} \\ \hline
\end{array}
\]
}
\caption{Property Graph Model}
\label{fig:pgm}
\end{figure}

\noindent {\bf Graph Query Language.} 
To query the above property graph model, we rely on a fragment of \emph{regular path queries} (RPQ) (\cite{CalvaneseDLV02}, \cite{ConsensM90}, \cite{CruzMW87}), which we enrich 
with aggregate operators, as depicted in Fig.~\ref{fig:graphg}. RPQs
correspond to property paths in SPARQL 1.1 and 
are a well-studied query class tailored to express \emph{label-constrained graph reachability patterns}, consisting of one or more
\emph{label-constrained reachability paths}. 
Given an alphabet $L_E$  of edge labels $l_e$ and vertices $v_1$ and $v_k$, the \emph{labeled path} $\pi$,
corresponding to $v_1 \xrightarrow{l^{1}_{e}} v_2 \ldots v_{k-1} \xrightarrow {l^{k}_{e}} v_k$, is the edge label concatenation $l^{1}_{e} \cdot \ldots \cdot l^{k}_{e}$.
In their full generality, RPQs allow to select vertices connected via such labeled paths that belong to a \emph{regular language} over $L_E$.
To our ends, we restrict RPQs to handle \emph{atomic paths} -- bi-directional, optional, single-labeled ($l_e$, $l_e?$, and $l_e^{-}$)
and transitive single-labeled ($l_e^{*}$) -- and \emph{composite paths} -- conjunctive and disjunctive 
composition of atomic paths ($l_e \cdot l_e$ and $\pi + \pi$). The expressivity of the identified fragment
is thus on par with that of Cypher.
While not as general as SPARQL, it retains relevance, since it, for example, already captures
more than 60\% of the property path SPARQL queries users write in practice \cite{BonifatiMT17}. 

\begin{figure}[t!]
{\small 
\[
\arraycolsep=1.2pt\def\arraystretch{1}
\begin{array}{clll}
\hline
\text{Clauses} ~ & C & ~{\color{red}::=}~ & A \leftarrow A_1, \ldots, A_n  ~|~ Q \leftarrow A_1, \ldots, A_n \\ \hline
\text{Queries}  ~ & Q & ~{\color{red}::=}~ & Ans(p_1,\ldots,p_m, count(x_1, \ldots, x_n)) \\ \hline
\text{Atoms}   ~ & A & ~{\color{red}::=}~ & \pi(l_{v_{1}}, l_{v_{2}}), \text{ for } l_{v_{1}}, l_{v_{2}} \in L_V ~|~ \leq(l_{v_{1}}, l_{v_{2}}) ~|~ \\ 
                  & & & <(l_{v_{1}}, l_{v_{2}}) ~|~ \geq(l_{v_{1}}, l_{v_{2}}) ~|~ >(l_{v_{1}}, l_{v_{2}}) \\ \hline
\text{Paths}  ~ & \pi & ~{\color{red}::=}~ & \epsilon ~|~ l_{e} ~|~ l_{e}? ~|~ l_{e}^{-1} ~|~ l_{e}^{*}, \text{ for } l_{e} \in L_E ~|~ 
                      \pi + \pi ~~|~~ l_e \cdot l_e \\
\hline
\end{array}
\]
}
\caption{Graph Query Language}
\label{fig:graphg}
\vspace{-1mm plus 1mm minus 1 mm}
\end{figure}

Moreover, this captures property path queries, as found in both the Wikidata online query collection
and the Wikidata large corpus studied in \cite{Malyshev18}. These are further enriched with the \emph{count} operator,
to support basic graph reachability estimates.

\begin{example}
We report in Figure \ref{ref:qrpq} the queries of Example \ref{ref:ads} expressed by using the syntax of Figure \ref{fig:graphg}.
\end{example}

\begin{figure}[t!]
\vspace{-2mm plus 2mm minus 2 mm}
\setlength\tabcolsep{2.5pt}
\begin{scriptsize}
\begin{tabular}{|c| l|} \hline
$Q_1(l_5)$ & $Ans(count(\_)) \leftarrow {\color{red}{l_5}}(\_,\_)$\\ \hline
$Q_2(l_2)$ & $Ans(count(\_))\leftarrow{\color{green}{l_2?}}(\_,\_)$\\ \hline
$Q_3(l_0)$ & $Ans(count(\_))\leftarrow l_0^{+}(\_,\_)$\\ \hline
$Q_4(l_0)$ & $Ans(count(\_))\leftarrow l_0^{*}(\_,\_)$\\ \hline
$Q_5(l_4,l_1)$ & $Ans(count(\_))\leftarrow{\color{orange}{l_4}} + {\color{blue}{l_1}}(\_,\_)$\\ \hline
$Q_6(l_4,l_5)$ & $Ans(count(\_))\leftarrow{\color{orange}l_4} \cdot {\color{red} l_5}(\_,\_)$\\ \hline
$Q_7(l_4,l_5)$ & $Ans(count(x))\leftarrow{\color{orange}l_4} \cdot {\color{red} l_5}, \geq(x.age,18), \leq(x.age,24)$.\\ \hline
\end{tabular}
\end{scriptsize}
\caption{Targeted Advertising Marker Queries}
\label{ref:qrpq}
\vspace{-3mm plus 2mm minus 0.1 mm}
\end{figure}

\section{Graph Summarization}\label{sec:gsum}
Let us assume a graph $\mathcal{G}=(V,E)$ and a edge label set $\Lambda_{Q} \subseteq \Lambda(\mathcal{G})$.
We introduce the \emph{GRASP summarization} algorithm, which compresses $\mathcal{G}$ into an AQP-amenable property graph, $\mathcal{\hat{G}}$,
tailored for counting label-constrained reachability queries (see Fig. \ref{fig:graphg}), whose labels are all in $\Lambda_{Q}$. 
Note that we place ourselves in the \emph{static} setting, as GRASP relies on transitive closure computation, whose maintenance under updates constitutes
an orthogonal problem to the one we target, as originally identified in \cite{PoutreL87}.

The GRASP summarization, described in Algorithm \ref{gsum}, consists of three phases. The \textbf{grouping phase} computes $\Phi$, 
a label-driven partitioning of $\mathcal{G}$ into \emph{groupings}, following the label-connectivity on the most frequently occurring edge 
labels in $\Lambda(\mathcal{G})$. Next, the \textbf{evaluation phase} refines the previous step, by further isolating into \emph{supernodes} 
the grouping components that satisfy a custom property concerning label-connectivity. The \textbf{merge phase} then fuses supernodes into 
\emph{hypernodes}, based on label-reachability similarity conditions, as specified by each heuristic mode $m$.

\begin{algorithm}[H]
\caption{GRASP($\mathcal{G}$, $\Lambda_{Q}$, $m$)}
\begin{flushleft}
 \textbf{Input:}  $\mathcal{G}$ -- a graph; $\Lambda_{Q} \subseteq \Lambda(\mathcal{G})$ -- a set of query labels;\\
                  \qquad \quad $\;m$ -- heuristic mode boolean switch \\
 \textbf{Output:} $\hat{\mathcal{G}}$ -- a graph summary
\end{flushleft}
\begin{algorithmic}[1]
\State $\Phi \leftarrow \text{GROUPING}(\mathcal{G})$
\vspace{1.5mm}
\State $\mathcal{G}^{*} \leftarrow \text{EVALUATION}(\Phi, \Lambda_{Q})$
\vspace{1.5mm}
\State $\hat{\mathcal{G}} \leftarrow \text{MERGE}(\mathcal{G}^{*}, \Lambda_{Q}, m)$
\vspace{1mm}
\State \Return $\hat{\mathcal{G}} = (\hat{V}, \hat{E})$
\end{algorithmic}
\label{gsum}
\end{algorithm}

The GRASP summarization phases are detailed in Sections \ref{sec:gphase}, \ref{sec:evalphase}, and \ref{sec:mergephase}.
To gauge the topological differences between realistic graph instances and their respective GRASP-summaries, 
we resort to visualization techniques in Section \ref{sec:vizg}.

\subsection{Grouping Phase}\label{sec:gphase}

The \textbf{grouping phase} aims to output a partitioning $\Phi$ of $\mathcal{G}$, such that $|\Phi|$ is \emph{minimized} and, 
for each $\mathcal{G}_i \in \Phi$, the number of occurrences of the most frequent edge label in $\Lambda(\mathcal{G}_i)$,
$\max\limits_{l \in \Lambda(\mathcal{G}_i)}(\#l)$, is \emph{maximized}. To this end, we first sort the 
set of edge labels in $\mathcal{G}$, $\Lambda(\mathcal{G})$, 
into a \emph{frequency list}, $\overrightarrow{\Lambda(\mathcal{G})}$. Next, for each $l_i \in \overrightarrow{\Lambda(\mathcal{G})}$, 
in descending frequency order, we set to identify the largest subgraphs of $\mathcal{G}$ that are 
weakly-connected on $l_i$. By relying on a \emph{most-frequently-occurring-label} heuristic, we thus bias the graph partitioning 
towards a coarser-level of granularity.

The key notion required to define $\Phi$ is that of \emph{maximal weak label-connectivity}, introduced below. 
Let us first introduce needed preliminary notions. In the following, we denote by $\overline{\mathcal{G}} = (V, \overline{E})$, 
where $|E| = |\overline{E}|$, the transformation of $\mathcal{G}$ into an undirected graph. Also, a 
graph $\mathcal{G}'$ is a \emph{subgraph} of $\mathcal{G}$ iff $V' \subseteq V$ and $E' \subseteq E$.

\begin{definition}[Weak Connectivity] $\mathcal{G}$ is \emph{weakly connected} iff $\overline{\mathcal{G}}$ is \emph{connected}, i.e.,
a path exists between any pair of $V$ vertices.
\end{definition}

\begin{definition}[Maximal Weak Connectivity] A subgraph of $\mathcal{G}$, $\mathcal{G}' = (V',E')$ is 
\emph{maximal weakly connected} iff: 1) $\mathcal{G}'$ is weakly connected and 2) no $E$ edge connects any $V'$ node to $V \setminus V'$.
\end{definition}

As we are interested in compressing $\mathcal{G}$ by label-connectivity, we strengthen the definition of weak-connectivity below.
%
\begin{definition}[Weak Label-Connectivity]
Given a graph $\mathcal{G} = (V, E)$ and a label $l \in \Lambda(\mathcal{G})$, $\mathcal{G}$
is \emph{weakly label-connected on $l$} iff: 
1) when converting all edges in $E$ into undirected ones, the resulting graph, $\overline{\mathcal{G}} = (V, \overline{E})$,
where $|E| = |\overline{E}|$, is \emph{connected} and
2) when removing any $l$-labeled edge from $\overline{E}$, there exist vertices $v_1, v_2$ in $V$ that are not connected
in $\overline{\mathcal{G}}$ by a $l^{+}$-labeled path.
\end{definition}

Since we aim to capture as many vertices in each \emph{weakly label-connected} subgraph of $\mathcal{G}$, we further enforce on the latter
the notion of \emph{maximality}, as defined next.

\begin{definition}[Maximal Weak Label-Connectivity]\label{def:mwlc}
Given a graph $\mathcal{G} = (V, E)$ and a label $l \in \Lambda(\mathcal{G})$, a subgraph of $\mathcal{G}$, $\mathcal{G'} = (V', E')$,
is \emph{maximal weakly label-connected} on $l$, denoted as $\lambda(\mathcal{G'}) = l$, iff: 1) $\mathcal{G'}$ is \emph{weakly label-connected on $l$} and
2) no edge in $E$, with label $l$, connects any of the nodes in $V'$ to $V \setminus V'$.
\end{definition}

We outline the \textbf{grouping phase}, as captured in Algorithm \ref{alg:group}.

\begin{algorithm}[H]
\begin{flushleft}
 \textbf{Input:}  $\mathcal{G}$ -- a graph\\
 \textbf{Output:} $\Phi$ -- a graph partitioning
\end{flushleft}
\caption{GROUPING($\mathcal{G}$)}
\begin{algorithmic}[1]
\LeftComment {Initialization}
\State $n \leftarrow |\Lambda(\mathcal{G})|$, $\overrightarrow{\Lambda(\mathcal{G})} \leftarrow [l_1, \ldots, l_{n}]$, $\Phi \leftarrow \emptyset$, $i \leftarrow 1$
\vspace{1.5mm}
\LeftComment {Label-driven partitioning computation}
\For {\textbf{all} $ l \in \overrightarrow{\Lambda(\mathcal{G})}$}
\State $\mathcal{G}_i \leftarrow \{\mathcal{G}^{k}_i = (V^{k}_i, E^{k}_i) \subseteq \mathcal{G} ~|~ \lambda(\mathcal{G}^{k}_i) = l\}$
\State $\Phi \leftarrow \Phi \cup \{\mathcal{G}_i\}$
\State $V \leftarrow V \setminus \bigcup\limits_{\mathcal{G}^{k}_i \in \mathcal{G}_i} \{v \in V ~|~ v \in V^{k}_i\}$
\State $i \leftarrow i+1$
\EndFor
\State $V(\Phi) \leftarrow \bigcup\limits_{\mathcal{G}^{\phantom{k}}_i \in \Phi} {\bigcup\limits_{\mathcal{G}^{k}_i \in \mathcal{G}_i} \{v \in V ~|~ v \in V^{k}_i\} }$
\State $\Phi \leftarrow \Phi \cup \{\mathcal{G}_i = (V_i, E_i) \subseteq \mathcal{G} ~|~ V_i = V \setminus V(\Phi) \}$
\State \Return $\Phi$
\end{algorithmic}
\label{alg:group}
\end{algorithm}

We henceforth denote $\Phi = \emph{GROUPING}(\mathcal{G})$ and name each $\mathcal{G}'$, $\mathcal{G}' \in \Phi$,
a $\mathcal{G}$-\emph{grouping} and each $\mathcal{G}''$, $\mathcal{G}'' \in \mathcal{G}'$, a $\mathcal{G}'$-\emph{subgrouping}. 
Note that $\Phi$ is not unique, as, for $l_1, l_2 \in \Lambda(\mathcal{G})$, such that $\#l_1 = \#l_2$, we arbitrarily order $l_1$ and 
$l_2$ in $\overrightarrow{\Lambda(\mathcal{G})}$. 

\begin{definition}[Non-Trivial (Sub)Groupings]
A $\mathcal{G}$-grouping, $\mathcal{G}'$, $\mathcal{G}' = (V', E')$, is called \emph{trivial}, if $\mathcal{G}' = \mathcal{G}$
or $E' = \emptyset$, and \emph{non-trivial}, otherwise. A $\mathcal{G}'$-subgrouping, $\mathcal{G}'' = (V'', E'')$, is called \emph{trivial}, if 
$E'' = \emptyset$, and \emph{non-trivial}, otherwise.
\end{definition}

\begin{lemma}[Non-Trivial Grouping Properties]
Let $\mathcal{G}'$ be a non-trivial $\mathcal{G}$-grouping. The following hold. 
\noindent{\textbf{P1:} For any non-trivial $\mathcal{G'}$-subgrouping, $\mathcal{G}''$, there exists $l''$, 
$l'' \in \Lambda(\mathcal{G'})$, such that $\lambda(\mathcal{G'}) = l''$}.
\noindent{\textbf{P2:}} For any non-trivial distinct $\mathcal{G'}$-subgroupings, $\mathcal{G}''_1$, $\mathcal{G}''_2$:
 a) $\lambda(\mathcal{G}''_1) = \lambda(\mathcal{G}''_2)$ and 
b) $\mathcal{G}''_1$ and $\mathcal{G}''_2$ are edge-wise disjoint.
\end{lemma}

\begin{proof}
$\textbf{P1}$ is provable by contradiction. If $\nexists l''$, $l'' \in \Lambda(\mathcal{G'})$, 
such that $\lambda(\mathcal{G'}) = l''$, then $E' = \emptyset$, contradicting the non-triviality of $\mathcal{G'}$.
\textbf{P2.a)} holds by construction and \textbf{P2.b)}, by contradiction. Assume 
$\mathcal{G}''_1 \cap \mathcal{G}''_2 \neq \emptyset$; then, $\mathcal{G}''_1$ and $\mathcal{G}''_2$
share at least a node, which is impossible by construction.  
\end{proof}

Next, let us characterize the \emph{GROUPING} algorithm. We start with the following observations.
%

\begin{lemma}[Subgrouping Maximal Label-Connectivity]\label{lemma:smlc}
For each $\mathcal{G}_i$, $\mathcal{G}_i \in \Phi$, each of its \emph{maximally weakly connected} components, $\mathcal{G}^{k}_i$, 
$\mathcal{G}^{k}_i \in \mathcal{G}_i$ is also \emph{maximally label-connected} on $l$, where $\#l = \max\limits_{l \in \Lambda(\mathcal{G}_i)}(\#l)$.
\end{lemma}
\begin{proof}
By construction, we know that, if $\mathcal{G}^{k}_i \in \mathcal{G}_i$, then there exists a label $l'$, $l' \in \Lambda(\mathcal{G})$, such 
that $\lambda(\mathcal{G}^{k}_i) = l'$. Assume that $l' \neq l$. By definition, there exists at least one edge in $E^{k}_i$ labeled $l$.
Since $\mathcal{G}^{k}_i$ is maximally weakly label-connected on $l'$, then each such edge has to connect vertices that are also connected
by an edge labeled $l'$. As $\#l \geq \#l'$, then there exists at least one pair of vertices in ${V}^{k}_i$ that are connected by
more edges labeled $l$ than $l'$. Hence, $\lambda(\mathcal{G}^{k}_i) = l$, contradicting the hypothesis.
\end{proof}

\begin{theorem}[\emph{GROUPING} Properties]\label{thm:grp}
If $|V| \geq 1$, then:\\
\noindent{\textbf{P1:} $\forall \mathcal{G}_{i} \in \Phi, V_i \neq \emptyset$}\\
\noindent{\textbf{P2:} $\forall \mathcal{G}_{i}, \mathcal{G}_{j} \in \Phi$, where $i \neq j$, $ \mathcal{G}_{i} \cap \mathcal{G}_{j} = \emptyset$}\\
\noindent{\textbf{P3:} $\bigcup\limits_{i \in [1,k]} V_{i} = V$ and $\bigcup\limits_{i \in [1,k]} E_{i} \subseteq E$}\\
\noindent{\textbf{P4:} $\Phi = \{\mathcal{G}_{i} = (V_i, E_i) \subseteq \mathcal{G} ~|~ i \in [1,|\Lambda(\mathcal{G})| + 1] \}$}
\end{theorem}
\begin{proof} $\textbf{P1}, \textbf{P2}, \textbf{P3}$ trivially hold.
Let us prove $\textbf{P4}$. If $E = \emptyset$, $\Phi = \{ \mathcal{G}\}$. Otherwise, a
label $l \in \overrightarrow{\Lambda(\mathcal{G})}$ exist and we can construct a grouping $\mathcal{G}_i$,
$\lambda(\mathcal{G}_i) = l$. Assume $\Phi > |\Lambda(\mathcal{G})| + 1$. At least two groupings, 
$\mathcal{G}_{i}, \mathcal{G}_{j}$, with the same most frequently 
occurring label, $l$, exist in $\Phi$. As $|\mathcal{G}_{i}| \geq 1$ and $|\mathcal{G}_{j}| \geq 1$, each contains
a maximally weakly connected component, $\mathcal{G}^{k_i}_{i}$, $\mathcal{G}^{k_j}_{j}$. 
From Lemma \ref{lemma:smlc}, $\lambda(\mathcal{G}^{k_i}_i) = \lambda(\mathcal{G}^{k_j}_j)$, contradicting
$\mathcal{G}_{i} \cap \mathcal{G}_{j} \neq \emptyset$.
\end{proof}
We illustrate the above algorithm, in Figure \ref{fig:M1}, as follows. 

\begin{example}[Graph Grouping]\label{ex:gg}
Let $\mathcal{G}$ be the graph from Figure \ref{ref:sn}. It holds that: 
$\#l_0 = 11, \, \#{\color{blue}l_1} = 3, \, \#{\color{green}l_2} = 2, \, \#{\color{purple}l_3} = 6, \, \#{\color{orange}l_4} = 7, \, \#{\color{red}l_5} = 7, \, \#{\color{brown}l_6} = 1$.
Hence, we can take $\overrightarrow{\Lambda(\mathcal{G})} = [l_0, {\color{red}l_5}, {\color{orange}l_4}, {\color{purple}l_3}, {\color{blue}l_1}, {\color{green}l_2}, {\color{brown}l_6}]$.
Note that, as $\#{\color{orange}l_4} = \#{\color{red}l_5}$, we can choose an arbitrary order between the labels in $\overrightarrow{\Lambda(\mathcal{G})}$.
Following Algorithm \ref{alg:group}, we first add $\mathcal{G}_1$ to $\Phi$, as it regroups the maximal weakly-label components on $l_0$. 
We then have $V = \{R_1,\ldots, R_7, M_1, \ldots, M_6, F_1, F_2\}$.
Next, we add $\mathcal{G}_2$ to $\Phi$, as it regroups the maximally weakly-label component on ${\color{red}l_5}$. 
We obtain $V = \{F_1, F_2\}$.
We add the remaining subgraph, $\mathcal{G}_3$, to $\Phi$ and output $\Phi = \{\mathcal{G}_1, \mathcal{G}_2, \mathcal{G}_3\}$, as illustrated 
in Figure \ref{fig:M1}.
\pgfdeclarelayer{background layer}
\pgfdeclarelayer{foreground layer}
\pgfsetlayers{background layer,main,foreground layer}
\tikzstyle{leg}=[rectangle, rounded corners, thin,
                       fill=gray!20, text=blue, draw]
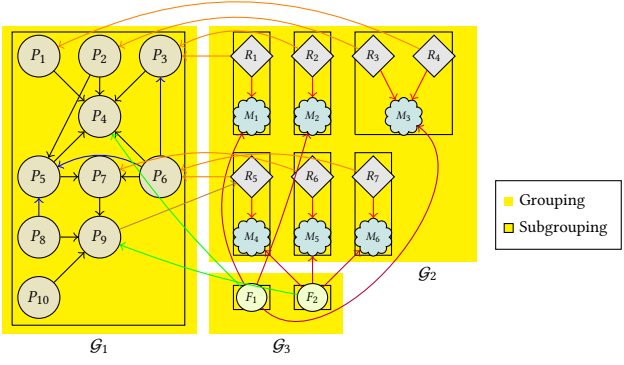
\begin{figure}[t!]
   \begin{tikzpicture}[scale=1.6, every node/.style={scale=0.7}]
   \tikzstyle{fe}=[ellipse, draw, thin,fill=lime!20, scale=0.8]
   \tikzstyle{mr}=[cloud, draw, thin,fill=teal!20, scale=0.7]
   \tikzstyle{rd}=[diamond, draw, thin,fill=gray!20, scale=0.8]
   \tikzstyle{mm}=[rectangle, maximum height=0cm, draw, thin]
   \tikzstyle{gr}=[rectangle, draw=none, thin,fill=yellow, scale=0.8]
   \tikzstyle{sn}=[rectangle, draw, thin,fill=yellow, scale=0.8]

  \begin{pgfonlayer}{background layer}
     \fill[yellow] (-0.3,0.2) rectangle (1.3,2.75);
     \fill[yellow] (1.4,0.2) rectangle (2.5,0.7);
     \fill[yellow] (3.6,0.8) rectangle (1.4,2.75);
    \draw (-0.22,0.27) -- (-0.22,2.7)-- (1.2,2.7) -- (1.2,0.27) --  (-0.22,0.27) ;
    \draw (1.6,1.85) -- (1.6,2.7) -- (1.9,2.7) -- (1.9,1.85) -- (1.6,1.85)   ;
    \draw (2.1,1.85) -- (2.1,2.7) -- (2.4,2.7) -- (2.4,1.85) -- (2.1,1.85)   ;
    \draw (2.6,1.85) -- (2.6,2.7) -- (3.4,2.7) -- (3.4,1.85) -- (2.6,1.85)   ;
    \draw (1.6,0.85) -- (1.6,1.7) -- (1.9,1.7) -- (1.9,0.85) -- (1.6,0.85)   ;
    \draw (2.1,0.85) -- (2.1,1.7) -- (2.4,1.7) -- (2.4,0.85) -- (2.1,0.85)   ;
    \draw (2.6,0.85) -- (2.6,1.7) -- (2.9,1.7) -- (2.9,0.85) -- (2.6,0.85)   ;
    \draw (1.6,0.4) -- (1.6,0.6)-- (1.9,0.6) -- (1.9,0.4) -- (1.6,0.4) ;
    \draw (2.1,0.4) -- (2.1,0.6)-- (2.4,0.6) -- (2.4,0.4) -- (2.1,0.4) ;

  \tikzstyle{every state}=[fill=olive!20,draw=black,text=black, inner sep=0.8pt]
    \node (G) at (0.5, 0.1) {$\mathcal{G}_1$};
    \node (G) at (2, 0.1) {$\mathcal{G}_3$};
    \node (G) at (3.2, 0.7) {$\mathcal{G}_2$};
    \node[state]         (P1)  at (0,2.5)    {$P_1$};
    \node[state]         (P2)  at (0.5,2.5)  {$P_2$};
    \node[state]         (P3)  at (1,2.5)    {$P_3$};
    \node[state]         (P4)  at (0.5,2)    {$P_4$};
    \node[state]         (P5)  at (0,1.5)    {$P_5$};
    \node[state]         (P6)  at (1,1.5)    {$P_6$};
    \node[state]         (P7)  at (0.5,1.5)  {$P_7$};
    \node[state]         (P8)  at (0,1)      {$P_8$};
    \node[state]         (P9)  at (0.5,1)    {$P_9$};
    \node[state]         (P10) at (0,0.5)    {$P_{10}$};
    \node[rd]         (R1)  at (1.75,2.5)    {$R_1$};
    \node[rd]         (R2)  at (2.25,2.5)    {$R_2$};
    \node[rd]         (R3)  at (2.75,2.5)    {$R_3$};
    \node[rd]         (R4)  at (3.25,2.5)    {$R_4$};
    \node[mr]         (M1)  at (1.75,2)      {$M_1$};
    \node[mr]         (M2)  at (2.25,2)      {$M_2$};
    \node[mr]         (M3)  at (3,2)         {$M_3$};
    \node[fe]            (F1)  at (1.75,0.5)    {$F_1$};
    \node[fe]            (F2)  at (2.25,0.5)  {$F_2$};
    \node[mr]         (M4)  at (1.75,1)   {$M_4$};
    \node[mr]         (M5)  at (2.25,1)   {$M_5$};
    \node[mr]         (M6)  at (2.75,1)   {$M_6$};
    \node[rd]         (R5)  at (1.75,1.5)   {$R_5$};
    \node[rd]         (R6)  at (2.25,1.5)   {$R_6$};
    \node[rd]         (R7)  at (2.75,1.5)   {$R_7$};

  \path (P1) edge[->]   node {}  (P4)
	 (P2) edge[->]   node {}  (P4)
	 (P3) edge[->]   node {}  (P4)
	 (P5) edge[->]   node {}  (P4)
	 (P6) edge[->]   node {}  (P4)
	 (P6) edge[->, color=blue]   node {}  (P3)
	 (P6) edge[->, bend right=28, above, color=blue]   node {}  (P5)
	 (P2) edge[->]   node {}  (P5)
	 (P5) edge[->]   node {}  (P7)
	 (P6) edge[->]   node {}  (P7)
	 (P7) edge[->]   node {}  (P9)
	 (P8) edge[->]   node {}  (P9)
	 (P8) edge[->, color=blue]   node {}  (P5)
	 (P10) edge[->]  node {}  (P9)
	 (R1) edge[->, color=red]   node {}  (M1)
	 (R2) edge[->, color=red]   node {}  (M2)
	 (R3) edge[->, color=red]   node {}  (M3)
	 (R4) edge[->, color=red]   node {}  (M3)
 	 (F1) edge[->, bend left=29, color=purple]   node {}  (M1)
 	 (F1) edge[->, bend right=3, color=purple]   node {}  (M2)
 	 (F1) edge[->, bend right=95, color=purple]   node {}  (M3)
 	 (F2) edge[->, color=purple]   node {}  (M4)
 	 (F2) edge[->, color=purple]   node {}  (M5)
 	 (F2) edge[->, color=purple]   node {}  (M6)
	 (R5) edge[->, color=red]   node {}  (M4)
	 (R6) edge[->, color=red]   node {}  (M5)
   	 (R7) edge[->, color=red]   node {}  (M6)
   	 (F1) edge[->, bend left=5, color=green]   node {}  (P4)
   	 (F2) edge[->, bend left=5, color=green]   node {}  (P9)
   	 (R1) edge[->, color=orange]   node {}  (P3)
   	 (R5) edge[->, color=orange]   node {}  (P6)
   	 (R6) edge[->, bend right=20, above, color=orange]   node {}  (P7)
   	 (R7) edge[->, bend right=19, above, color=orange]   node {}  (P6)
   	 (R2) edge[->, bend right=28, above, color=orange]   node {}  (P3)
   	 (R3) edge[->, bend right=26, above, color=orange]   node {}  (P2)
   	 (R4) edge[->, bend right=26, above, color=orange]   node {}  (P1)
   	 (P9) edge[->, color=brown]   node {}  (R5)
   ;
   \end{pgfonlayer}

   \matrix(table) [draw,below right] at (current bounding box.east) {
     \node [gr,label=right:Grouping]  {}; \\
     \node [sn,label=right:Subgrouping] {}; \\
      };
  \end{tikzpicture}
\caption{Summarizing $\mathcal{G}_{SN}$ (Grouping Phase)} \label{fig:M1}
\end{figure}
\end{example}
%
%
%
%
\subsection{Evaluation Phase}\label{sec:evalphase}

The \textbf{evaluation phase} takes as input $\Phi$, the $\mathcal{G}$-partitioning from Algorithm \ref{alg:group},
and $\Lambda_Q$, a set of labels, and outputs $\mathcal{G}^{*} = (V^{*}, E^{*})$, an AQP-amenable compression of $\mathcal{G}$. 
The phase creates $V^{*}$, the set of \emph{supernodes} (SN), 
and $E^{*}$, the set of \emph{superedges} (SE). 
After each step, $\mathcal{G}^{*}$ is enriched with AQP-relevant properties, 
exploited in Section \ref{sec:aqp}. Next, we explain how \underline{\emph{supernodes}}, \underline{\emph{AQP-properties}}, and 
\underline{\emph{superedges}} are computed. 

\begin{definition}[Supernodes (SN)]\label{def:sn}
Let $\Phi$ be a $\mathcal{G}$-partitioning into groupings, $\mathcal{G}_i$.
$\Phi$ is transformed into a set of supernodes, $V^{*} = VFUSE(\Phi)$ (see Algorithm \ref{alg:vf}).
Each \emph{supernode}, $v^{*} \in V^{*}$, is obtained by fusing all vertices and edges of each subgrouping $\mathcal{G}^{k}_i$,
$\mathcal{G}^{k}_i \in \mathcal{G}_i$. We denote the label $l$, such that $\lambda(\mathcal{G}^{k}_i) = l$, as $\lambda(v^{*})$.
\end{definition}

\begin{algorithm}[H]
\begin{flushleft}
 \textbf{Input:}  $\Phi$ -- a graph partitioning;\\
 \textbf{Output:} $V^{*}$ -- set of \emph{supernodes}
\end{flushleft}
\begin{algorithmic}[1]
\caption{VFUSE($\Phi$)}
\State $V^{*} \leftarrow \emptyset$
\For {\textbf{all} $ \mathcal{G}_{i} \in \Phi$}
\For {\textbf{all} $ \mathcal{G}^{k}_{i} \in \mathcal{G}_{i}$}
\State $v^{*}_{k} \leftarrow \mathcal{G}^{k}_{i}$
\State $V^{*} \leftarrow V^{*} \cup \{v^{*}_{k}\}$
\EndFor
\EndFor
\State \Return $V^{*}$
\end{algorithmic}
\caption{}
\label{alg:vf}
\end{algorithm}

\noindent \underline{\emph{AQP-properties}}. Consider $l \in \Lambda(\mathcal{G})$ and supernodes $v^{*}_i, v^{*}_j \in V^{*}$. 
We call a $l$-labeled edge from $v^{*}_i$ to $v^{*}_j$ a \emph{cross-edge} and define 
$E_{i,j}(l) = \{e \in E | e.1 \in v^{*}_i \wedge e.2 \in v^{*}_j \wedge \gamma(e) = l\}$. 
For every $v^{*} \in \mathcal{G}^{k}_{i}$, we then associate $\sigma(v^{*})$, a set consisting of properties that regard \emph{compression},
\emph{label-connectivity}, and \emph{pairwise label-traversal}. We explain each of these below. \\
%
\noindent{\textbf{Compression}. We record the number of \emph{inner vertices} in $v^{*}$ as $\emph{VWeight}(v^{*}): |V^{k}_i|$ 
and that of \emph{inner edges} as $\emph{EWeight}(v^{*}): |E^{k}_i|$.\\
\noindent{\textbf{Label-Connectivity}. The percentage-wise occurrence 
of $l$ in $v^{*}$ is $\emph{LPercent}(v^{*}, l): \frac{|\{e \, \in \, E^{k}_i|~ \gamma(e) \, = \, l|}{\emph{EWeight}(v^{*})}$.
The number of vertex pairs connected with an $l$-labeled edge is 
$\emph{LReach}(v^{*}, l): |\{(v_1, v_2) \in V^{k}_i \times V^{k}_i|~ l^{+}(v_1, v_2) \in \mathcal{G}^{k}_{i}\}|$.\\
\noindent{\textbf{Pairwise Label-Traversal}. Consider labels $l_1, l_2 \in \Lambda(\mathcal{G})$ and direction indices, $d_1, d_2 \in \{1,2\}$.
We compute the number of paths between two cross-edges with labels $l_1$, $l_2$, 
directions $d_1$, $d_2$, and a common node, as $\emph{EReach}(v^{*}, l_1, l_2, d_1, d_2) = |\{(e^{*}_1, e^{*}_2) | \{e^{*}_1, e^{*}_2\} \subseteq E^{*}\setminus E^{k}_i
\wedge \gamma(e^{*}_1) = l_1 \wedge \gamma(e^{*}_2) = l_2 \wedge e^{*}_1.d_1 = v^{*} = e^{*}_2.d_2\}|$. We compute the number of \emph{traversal edges}, i.e., inner/cross-edge pairs, $e_1, e_2$, with labels $l_1$, $l_2$, directions
$d_1, d_2$, and $v^{*}$ as common endpoint, as $\delta(v^{*}, l_1, l_2, d_1, d_2) = 
|\{(e_1, e_2) |~ e_1 \in E^{*} \setminus E^{k}_i \wedge \gamma(e_1) = l_1 \wedge e_2 \in E^{k}_i \wedge \gamma(e_2) = l_2 \wedge
e_1.d_1 = v^{*} = e_2.d_2\}|$. We take the number of \emph{frontier vertices}, given a fixed label $l$ and direction $d$, to be
$V_F(v^{*}, l, d) = \{v|~ v \in v^{*} \wedge \exists e, e \in E \setminus E^{k}_i \wedge \gamma(e) = l \wedge e.d = v\}|$.
Finally, we consider the number of $l_1$-labeled cross-edges relative to that of frontier vertices on $l_2$ to represent 
the \emph{relative label participation}, computed as
$\emph{RLPart}(v^{*}, l_1, l_2, d_1, d_2): (\sum\limits_{v \in V^{k}_i}\delta(v, l_1, d_1, d_2, V \setminus V^{k}_i))/|V_F(v^{*}, l_2, d_2)|$.
\setlength{\textfloatsep}{0pt}
\begin{algorithm}[H]
\begin{flushleft}
 \textbf{Input:} $V$ -- a set of vertices, $\Lambda$ -- a set of labels;\\
 \textbf{Output:} $V$ -- property-enriched set of vertices
\end{flushleft}
\begin{algorithmic}[1]
\caption{VProperties($\Phi$, $\Lambda$)}
\For {\textbf{all} $v \in V$}
\State $v.vweight \leftarrow VWeight(v)$, 
\State $v.eweight \leftarrow EWeight(v)$
\For {\textbf{all} $(l_1,l_2) \in \Lambda$, $(d_1, d_2) \in \{1,2\}$}
\State $v.plabel(l_1) \leftarrow LPercent(v,l_1)$
\State $v.lreach(l_1) \leftarrow LReach(v,l_1)$
\State {$v.ereach(l_1, d_1, d_2) \leftarrow EReach(v,l_1,d_1,d_2)$}
\State {$v.rlpart(l_1, l_2, d_1, d_2) \leftarrow RLReach(v,l_1, l_2, d_1,d_2)$}
\EndFor
\EndFor
\State \Return $V$
\end{algorithmic}
\label{alg:ealg}
\end{algorithm}
We illustrate these properties via an example. 
\begin{example}[Supernode Properties]\label{snprop}
In Fig.\ref{fig:reach}, we have that $EReach(v^{*}_1, l_1, l_2, 1, 1)= \delta(v^{*}_2, l_1, l_2, 1, 2)= 1$,
$V_F(v^{*}_2, l_1, 1) = \{v_1, v_3\}$. 
 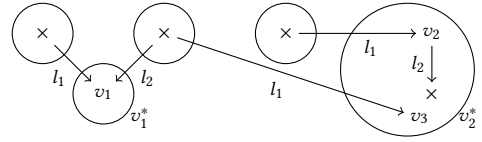
\begin{figure}[t!]
 \begin{center}
 \begin{tikzpicture}[scale=1.6, every node/.style={scale=0.8}]
  \draw[] (1,1) circle (0.25cm);
  \draw[] (0.5,1.5) circle (0.25cm);
  \draw[] (1.5,1.5) circle (0.25cm);
  \draw[] (2.5,1.5) circle (0.25cm);
  \draw[] (3.5,1.2) circle (0.55cm);
      \node[] (N1)  at (1,1)       {$v_1$};
      \node[] (N2)  at (0.5,1.5)   {$\times$};
      \node[] (N3)  at (1.5,1.5)   {$\times$};
      \node[] (N4)  at (2.5,1.5)   {$\times$};
      \node[] (N5)  at (3.7,1)     {$\times$};
      \node[] (N6)  at (3.7,1.5)   {$v_2$};
      \node[] (N7)  at (3.6,0.8)   {$v_3$};
      \node[] (N8)  at (1.3,0.8)   {$v^{*}_1$};
      \node[] (N9)  at (4,0.8)   {$v^{*}_2$};
 \path 
     (N3) edge[->] node[below right,scale=1]{${\color{black}l_2}$} (N1)
     (N2) edge[->] node[below left,scale=1]{${\color{black}l_1}$} (N1)
     (N4) edge[->] node[below right,scale=1]{${\color{black}l_1}$} (N6)
     (N6) edge[->] node[left,scale=1]{${\color{black}l_2}$} (N5)
     (N3) edge[->] node[below left,scale=1]{${\color{black}l_1}$} (N7)
     ;
 \end{tikzpicture}
\vspace{-3mm}
\end{center}
\caption{Traversal Nodes/Edges and Frontier Vertices}
\label{fig:reach}
\end{figure}
\end{example}
%
%

%
We now proceed to explaining the creation of \underline{\emph{superedges}}.

\begin{definition}[Superedges (SE)]
A \emph{superedge}, $e^{*} \in E^{*}$, is obtained by merging all cross-edges $e$, $e \in E_{i,j}(l)$,
$l \in \Lambda(\mathcal{G})$, as described in the below algorithm. To each such $e^{*}$ we then associate a weight, 
$\emph{EWeight}(e^{*}): |\{e \in E ~|~ e \in e^{*}\}|$. 
\end{definition}
\setlength{\textfloatsep}{0pt}
\begin{algorithm}[H]
\begin{flushleft}
 \textbf{Input:}  $\Phi$ -- a graph partitioning, $\Lambda_{Q}$ -- a set of label pairs;
 \textbf{Output:} $\mathcal{G}^{*}$ -- a graph with \emph{supernodes} and \emph{superedges}
\end{flushleft}
\begin{algorithmic}[1]
\caption{EFUSE($V^{*}, \Lambda$)}
\State $E^{*} \leftarrow \emptyset$
\For {\textbf{all} $ l \in \Lambda$}
\State $E_l \leftarrow \{e \in E | \gamma(e) = l\}$
\For {\textbf{all} $v^{*}_i, v^{*}_j \in V^{*}$}
\For {\textbf{all} $v_i \in v^{*}_i, v_j \in v^{*}_j$}
\State $e^{*} \leftarrow \{l(v^{*}_i, v^{*}_j) | l(v_i, v_j) \in E_l\}$
\State $E^{*} \leftarrow E^{*} \cup \{e^{*}\}$
\EndFor
\EndFor
\EndFor
\State \Return $E^{*}$
\end{algorithmic}
\label{alg:ealg}
\end{algorithm}

The \textbf{evaluation phase} is summarized next.
\setlength{\textfloatsep}{0pt}
\begin{algorithm}[H]
\begin{flushleft}
 \textbf{Input:}  $\Phi$ -- a graph partitioning, $\Lambda$ -- a set of labels;\\
 \textbf{Output:} $\mathcal{G}^{*}$ -- a graph with \emph{supernodes} and \emph{superedges}
\end{flushleft}
\begin{algorithmic}[1]
\caption{EVALUATE($\Phi$, $\Lambda$)}
\State $V^{*} \leftarrow VFUSE(\Phi)$
\State $V^{*} \leftarrow VProperties(V^{*}, \Lambda)$
\State $\hat{E} \leftarrow EFUSE(V^{*}, \Lambda)$
\For {\textbf{all} $e^{*} \in E^{*}$}
\State $e^{*}.weight \leftarrow EWeight(e^{*})$
\EndFor
\State \Return $\mathcal{G}^{*} = (V^{*}, E^{*})$
\end{algorithmic}
\label{alg:ealg}
\end{algorithm}

We illustrate the above by revisiting our running example. 

\begin{example}[Graph Compression]

In Figure \ref{fig:M2}, we display the supergraph $\mathcal{G}^{*}$, obtained from $\Phi$, after the
\textbf{evaluation phase}. Each supernode corresponds to a $\Phi$-subgrouping, as
depicted in Figure \ref{fig:M1}. Each superedge is obtained by compressing similarly labeled edges, 
whose source and, respectively, target vertices are in the same supernode. Each superedge $e^{*}$
has $EWeight(e^{*}) = 1$, except that connecting $SN_4$ and $SN_1$, whose edge weight is 2.
\noindent
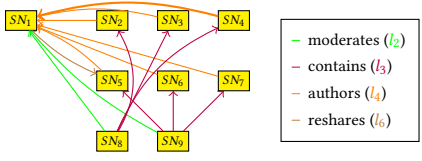
\begin{figure}[t!]
\begin{center}
   \begin{tikzpicture}[scale=1.6, every node/.style={scale=0.7}]
   \tikzstyle{sn}=[rectangle, draw, thin,fill=yellow, scale=0.8]
   \tikzstyle{every state}=[fill=olive!20,draw=black,text=black, inner sep=0.8pt]
    \node[sn]         (SN1)  at (3,2.5)  {$SN_1$};
    \node[sn]         (SN2)  at (3.75,2.5) {$SN_2$};
    \node[sn]         (SN3)  at (4.25,2.5) {$SN_3$};
    \node[sn]         (SN4)  at (4.75,2.5) {$SN_4$};
    \node[sn]         (SN5)  at (3.75,2) {$SN_5$};
    \node[sn]         (SN6)  at (4.25,2) {$SN_6$};
    \node[sn]         (SN7)  at (4.75,2) {$SN_7$};
    \node[sn]         (SN8)  at (3.75,1.5) {$SN_8$};
    \node[sn]         (SN9)  at (4.25,1.5) {$SN_9$};

   \path  (SN4) edge[thick, ->, bend right=15, above, color=orange] node {}  (SN1)
   (SN3) edge[->, bend right=13, above, color=orange] node {}  (SN1)
   (SN2) edge[->, color=orange] node {}  (SN1)
   (SN5) edge[->, color=orange, bend right=10] node {}  (SN1)
   (SN6) edge[->, color=orange] node {}  (SN1)
   (SN7) edge[->, color=orange] node {}  (SN1)
   (SN8) edge[->, color=green] node {}  (SN1)
   (SN9) edge[->, bend left=10, color=green] node {}  (SN1)
   
   (SN8) edge[->, bend right=30, color=purple] node {}  (SN2)
   (SN8) edge[->, color=purple] node {}  (SN3)
   (SN9) edge[->, color=purple] node {}  (SN5)
   (SN9) edge[->, color=purple] node {}  (SN6)
   (SN9) edge[->, color=purple] node {}  (SN7)
   (SN8) edge[->, bend left=20, color=purple] node {}  (SN4)
   (SN1) edge[->, color=brown, bend right=10] node {}  (SN5)
   ;     
    \hspace{4mm}
       \matrix(table) [draw,right,yshift=-5pt] at (current bounding box.east) {
      & \node [state,opacity=0,text opacity=1, label=right:moderates ({\color{green}$l_2$}),scale=0.3] {{\color{green}\par\noindent\rule{15pt}{1pt}}}; \\
      ;
        & \node [state,opacity=0,text opacity=1, label=right:contains ({\color{purple}$l_3$}),scale=0.3] {{\color{purple}\par\noindent\rule{15pt}{1pt}}}; \\
        & \node [state,opacity=0,text opacity=1, label=right:authors ({\color{orange}$l_4$}),scale=0.3]  {{\color{orange}\par\noindent\rule{15pt}{1pt}}}; \\
        & \node [state,opacity=0,text opacity=1, label=right:reshares ({\color{brown}$l_6$}),scale=0.3]  {{\color{brown}\par\noindent\rule{15pt}{1pt}}}; \\
      };
  \end{tikzpicture}
 \vspace{1.5mm}
\caption{Summarizing $\mathcal{G}_{SN}$ (Evaluation Phase)} \label{fig:M2}
\end{center}
\end{figure}
\end{example}
\subsection{Merge Phase}\label{sec:mergephase}

The \textbf{merge phase} takes as input a graph, $\mathcal{G}^{*}$, as computed by Algorithm \ref{alg:merge},
along with a set of labels, $\Lambda_Q$, and outputs a compressed graph, 
$\hat{\mathcal{G}} = (\hat{V}, \hat{E})$. The phase proceeds in two steps, corresponding to the creation of $\hat{V}$, 
the set of \underline{\emph{hypernodes}} (HN), and, respectively, to that of $\hat{E}$, the set of \underline{\emph{hyperedges}} (HE). 
At each step, $\hat{G}$ is enriched with \underline{\emph{AQP-relevant properties}} (see Section \ref{sec:aqp}).

\emph{Hypernodes} are computed by merging together supernodes based on various criteria, according to Definition \ref{def:hn}.
The primary, \emph{inner-merge}, condition for merging candidate supernodes is for them to be maximal weakly label-connected on the same label. 
The \emph{source-merge} heuristic additionally requires that they share the same set of outgoing labels,
while the \emph{target-merge} heuristic requires that they share the same set of ingoing labels. 
\begin{definition}[Hypernodes (HN)]\label{def:hn}
Let $V^{*}$ be a supernode set, where $V^{*} = \{v^{*}_1, \ldots, v^{*}_n\}$. 
$V^{*}$ is transformed into a set of hypernodes, $\hat{V} = VMERGE(V^{*})$,
where $\hat{V} = \{\hat{v}_1, \ldots, \hat{v}_m\}$. A \emph{hypernode} $\hat{v}_k \in \hat{V}$
corresponds to fusing a subset $V^{*}_k$ of $V^{*}$, such that, for all $i,j \in [1,n]$,
$\{v^{*}_i, v^{*}_j\} \in V^{*}_k$, if $\lambda(v^{*}_i) = \lambda(v^{*}_j)$, and either of the following holds:
{\noindent{\textbf{Case 1.}} $\Lambda_{+}(v^{*}_i) = \Lambda_{+}(v^{*}_j)$, for all $i,j \in [1,n]$.}
{\noindent{\textbf{Case 2.}} $\Lambda_{-}(v^{*}_i) = \Lambda_{-}(v^{*}_j)$, for all $i,j \in [1,n]$.}
The first condition defines the \emph{target-merge} heuristic, while the latter, the \emph{source-merge} one. 
\end{definition}

To each such hypernode, $\hat{v}$, we associate a property set.

\begin{definition}[Hypernode Properties] \label{def:hn}
The \emph{property set} of $\hat{v}$, $\sigma(\hat{v})$ is the same as in Definition \ref{def:sn}.
Except for LPercent, the values of hypernode properties are obtained by adding all those corresponding to \emph{supernodes} $v^{*}$,
such that $v^{*} \in \hat{v}$. For the label percentage property on a given label $l$, we have:
$\emph{LPercent}(\hat{v}, l): \frac{\sum\limits_{v^{*} \in \hat{v}} LPercent(v^{*},l) \,*\, EWeight(v^{*})}{\sum\limits_{v^{*} \in \hat{v}} 
EWeight(v^{*})}$.
\end{definition}

\emph{Superedges} are merged into \emph{hyperedges}, if they share the same label and endpoints, as defined below.

\setlength{\textfloatsep}{0pt}
\begin{algorithm}[H]
\begin{flushleft}
 \textbf{Input:}  $V^{*}$ -- a set of supernodes; $\Lambda$ -- a set of labels;\\
                  \qquad \quad$\;m$ -- heuristic mode \\
 \textbf{Output:} $\hat{V}$ -- a set of hypernodes
\end{flushleft}
\begin{algorithmic}[1]
\caption{VMERGE($V^{*}$, $\Lambda$, $m$)}
\For {\textbf{all} $v^{*} \in V^{*}$}
\State{\small$\Lambda_{+}(v^{*})\leftarrow\{l \in \Lambda ~|~ \exists v^{*}_s, v^{*}_s \in V^{*} \wedge l(v^{*}_s,v^{*}) \in E^{*} \}$}
\State{\small$\Lambda_{-}(v^{*})\leftarrow\{l \in \Lambda ~|~ \exists v^{*}_t, v^{*}_t \in V^{*} \wedge l(v^{*},v^{*}_t) \in E^{*} \}$}
\EndFor
\For {\textbf{all} $v^{*}_1, v^{*}_2 \in V^{*}$}
\State $b_{\lambda} \leftarrow \lambda(v^{*}_1)= \lambda(v^{*}_2)$ \Comment {Inner-Merge Condition}
\State $b_{+} \leftarrow \Lambda_{+}(v^{*}_1) = \Lambda_{+}(v^{*}_2)$, $b_{-} \leftarrow \Lambda_{-}(v^{*}_1) = \Lambda_{-}(v^{*}_2)$
\If {$m = \texttt{true}$}
\State $\hat{v} \leftarrow \{v^{*}_1, v^{*}_2 ~~|~~ b_{\lambda} \wedge b_{+} = \texttt{true} \}$ \Comment {\color{blue}{Target-Merge}}
\Else {}
\State $\hat{v} \leftarrow \{v^{*}_1, v^{*}_2 ~~|~~ b_{\lambda} \wedge b_{-} = \texttt{true} \}$ \Comment {\color{blue}{Source-Merge}}
\EndIf
\EndFor
\State $\hat{V} \leftarrow \{\hat{v}_{k} ~~|~~ k \in [1, |V^{*}|] \}$
\State \Return $\hat{V}$
\end{algorithmic}
\caption{}
\label{alg:merge}
\end{algorithm}

\begin{definition}[Hyperedges (HE)] \label{def:he}
Let $E^{*}$ be a superedge set, $E^{*} = \{e^{*}_1, \ldots, e^{*}_n\}$. 
$E^{*}$ is transformed into a hyperedge set, $\hat{E} = EMERGE(E^{*})$,
where $\hat{E} = \{\hat{e}_1, \ldots, \hat{e}_m\}$. A \emph{hyperedge}, $\hat{e}_k$, $\hat{e}_k \in \hat{e}$,
corresponds to fusing a subset $E^{*}_k$ of $E^{*}$, such that, for $l \in \Lambda(\mathcal{G})$ and
$i,j \in [1,n]$, $\{e^{*}_i, e^{*}_j\} \subseteq E^{*}_k$, if
$e^{*}_i \stackrel{l}{=} e^{*}_j$, $e^{*}_i.1 = e^{*}_j.1$, and $e^{*}_i.2 = e^{*}_j.2$.
To each hyperedge $\hat{e}$ we associate a weight, 
$\emph{EWeight}(\hat{e})$, corresponding to $|E^{*}_k|$, the number of superedges $\hat{e}_k$ compressed.
\end{definition}

The \textbf{merge phase} is captured in Algorithm \ref{fmerge}.

\begin{algorithm}[H]
\begin{flushleft}
 \textbf{Input:}  $\mathcal{G}^{*}$ -- a graph; $\Lambda_{Q}$ -- a set of label pairs; \\
                  \qquad \quad$\;m$ -- heuristic mode \\
 \textbf{Output:} $\hat{\mathcal{G}}$ -- a graph summary
\end{flushleft}
\begin{algorithmic}[1]
\caption{MERGE($\mathcal{G}^{*}$, $\Lambda$, $m$)}

\State $\hat{V} \leftarrow VMERGE(V^{*},\Lambda,m)$
\State $\hat{V} \leftarrow VProperties(\hat{V}, \Lambda)$
\State $\hat{E} \leftarrow EFUSE(E^{*}, \Lambda(\mathcal{G}^{*}))$
\For {\textbf{all} $\hat{e} \in \hat{E}$}
\State $\hat{e}.eweight \leftarrow EWeight(\hat{e})$
\EndFor
\State \Return $\hat{\mathcal{G}} = (\hat{V}, \hat{E})$
\end{algorithmic}
\caption{}
\label{fmerge}
\end{algorithm}

Finally, we depict the resulting GRASP-summarization of our running example, as follows . 

\begin{example}[Graph Compression]

The graphs in Figure \ref{fig:M3} are obtained from $\mathcal{G}^{*} = (V^{*},E^{*})$, after the
\textbf{merge phase}. Each hypernode corresponds to the fusion of the supernodes in Figure \ref{fig:M2},
by the heuristics target-merge (left) and source-merge (right).
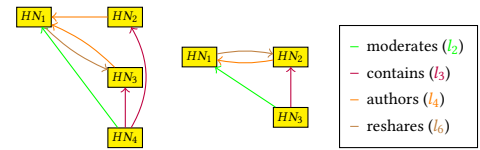
\begin{figure}[h!]
\begin{center}
   \begin{tikzpicture}[scale=1.6, every node/.style={scale=0.7}]
   \tikzstyle{sn}=[rectangle, draw, thin,fill=yellow, scale=0.8]
   \tikzstyle{every state}=[fill=olive!20,draw=black,text=black, inner sep=0.8pt]
    \node[sn]         (SN1)  at (3,2.5)  {$HN_1$};
    \node[sn]         (SN2)  at (3.75,2.5) {$HN_2$};
    \node[sn]         (SN5)  at (3.75,2) {$HN_3$};
    \node[sn]         (SN8)  at (3.75,1.5) {$HN_4$};

   \path  
   (SN2) edge[->, color=orange] node {}  (SN1)
   (SN5) edge[->, bend right=10, color=orange] node {}  (SN1)
   (SN8) edge[->, color=green] node {}  (SN1)
   (SN8) edge[->, bend right=30, color=purple] node {}  (SN2)
   (SN8) edge[->, color=purple] node {}  (SN5)
   (SN1) edge[->, bend right=10, color=brown] node {}  (SN5)
   ;
  \end{tikzpicture}
   \quad 
   \begin{tikzpicture}[scale=1.6, every node/.style={scale=0.7}]
   \tikzstyle{sn}=[rectangle, draw, thin,fill=yellow, scale=0.8]
   \tikzstyle{every state}=[fill=olive!20,draw=black,text=black, inner sep=0.8pt]
    \node[sn]         (SN1)  at (3,2.5)  {$HN_1$};
    \node[sn]         (SN2)  at (3.75,2.5) {$HN_2$};
    \node[sn]         (SN5)  at (3.75,2) {$HN_3$};
 
   \path  
   (SN2) edge[->, color=orange, bend left=10] node {}  (SN1)
   (SN1) edge[->, color=brown, bend left=10] node {}  (SN2)
   (SN5) edge[->, color=purple] node {}  (SN2)
   (SN5) edge[->, color=green] node {}  (SN1)
   ;
   
         \hspace{4mm}
       \matrix(table) [draw,right] at (current bounding box.east) {
      & \node [state,opacity=0,text opacity=1, label=right:moderates ({\color{green}$l_2$}),scale=0.3] {{\color{green}\par\noindent\rule{15pt}{1pt}}}; \\
      ;
        & \node [state,opacity=0,text opacity=1, label=right:contains ({\color{purple}$l_3$}),scale=0.3] {{\color{purple}\par\noindent\rule{15pt}{1pt}}}; \\
        & \node [state,opacity=0,text opacity=1, label=right:authors ({\color{orange}$l_4$}),scale=0.3]  {{\color{orange}\par\noindent\rule{15pt}{1pt}}}; \\
        & \node [state,opacity=0,text opacity=1, label=right:reshares ({\color{brown}$l_6$}),scale=0.3]  {{\color{brown}\par\noindent\rule{15pt}{1pt}}}; \\
      };
  \end{tikzpicture}
 \vspace{1.5mm}
\caption{Summarizing $\mathcal{G}_{SN}$ (Merge Phase)} \label{fig:M3}
\end{center}
\end{figure}
\end{example}

Finally, we show the intractability of the graph summarization problem, under the conditions of our algorithm. 

\begin{definition}[Summarization Function]\label{fsummary}
Let $\mathcal{G} = (V,E)$ be a graph and $\Phi =\{\mathcal{G}_{i} = (V_{i}, E_{i}) \,|\, i \in [1,|V|]\}$, 
a $\mathcal{G}$-partitioning into HNs. Each HN, $\mathcal{G}_{i}$, contains HN-subgraphs, $\mathcal{G}^{k}_{i}$,
all \emph{maximal weakly label-connected} on a label $l \in \Lambda(\mathcal{G})$.
A \emph{summarization function} $\chi_{\Lambda}: V \rightarrow \mathbb{N}$ is a function 
assigning to each vertex, $v$, a unique HN identifier $\chi_{\Lambda}(v) \in [1,k]$.
$\chi_{\Lambda}$ is \emph{valid} if for any $v_1, v_2$, where $\chi_{\Lambda}(v_1) = \chi_{\Lambda}(v_2)$, either $v_1, v_2$ belong to :\\
{\noindent\textbf{Case 1} the same HN-subgraph, $\mathcal{G}^{k}_{i}$, that is maximal weak label-connected on $l$, or to}\\
{\noindent\textbf{Case 2} different HN-subgraphs, $\mathcal{G}^{k_1}_{i}$, $\mathcal{G}^{k_2}_{i}$, each maximal label-connected on $l$ and 
not connected by an $l$-labeled edge in $\mathcal{G}$.}
\end{definition}

\begin{theorem}[Optimal Summarization NP-completeness]
Let \textbf{MinSummary} be the problem that, given a graph $\mathcal{G}$ and an integer $k' \geq 2$,
decides whether there exists a label-driven partitioning $\Phi$ of $\mathcal{G}$ with $|\Phi| \leq k'$,
such that $\chi_{\Lambda}$ is a \emph{valid summarization}. Then, \textbf{MinSummary} is NP-complete, even for undirected
graphs, $|\Lambda(\mathcal{G})| \leq 2$ and $k'=2$.
\begin{proof} We establish the result in two steps.\\
{\noindent\textbf{MinSummary} is in NP. We construct a valid \emph{summarization function}, $\chi_{\Lambda}$, as a witness.
 For a graph partitioning in $k$ \emph{subgraphs}, one can verify in polynomial time (see \cite{JinCLR10}) if two vertices 
 are reachable by a given labeled-constrained
 path and decide if their assignation to the same or to different HNs is valid (Def. \ref{fsummary}).}\\
{\noindent\textbf{MinSummary} is NP-hard. Let us henceforth reduce the \textbf{MinSummary} problem to \textbf{IndSet}, i.e., 
 the NP-complete (see \cite{GareyJ79}) problem of establishing whether an undirected graph contains $K$ independent vertices, for an arbitrary $K$.} 
{\noindent We prove \textbf{IndSet} $\leq_{p}$ \textbf{MinSummary}. Let $\mathcal{G} = (V,E)$ be an \textbf{IndSet} instance, 
where $\mathcal{G}$ is undirected, $|V| = n \geq 2$, $|E| = m$, 
 $\Lambda(\mathcal{G}) = \{{\color{blue}l_1}\}$. We consider a polynomial reduction function, $f$, such that 
 $f(\mathcal{G}) = \mathcal{G}'$, $\mathcal{G}' = (V', E')$ (see Fig. \ref{gadget}), $\{v'_1, v'_2, v'_3\} \subset V'$, 
 $\Lambda(\mathcal{G}) = \{{\color{blue}l_1}, {\color{red}l_2}\}$, and $\tilde{\mathcal{G}} \subset \mathcal{G}$, where
 $\tilde{\mathcal{G}}$ is obtained from $\mathcal{G}$, by adding, 
 between each pair of vertices connected with an ${\color{blue}l_1}$-labeled edge, $n$ more ${\color{blue}l_1}$-labeled edges.
 Let $\mathcal{G}'$ contain three paths of length $k$, between $v'_1$ and $v'_2$ (one, ${\color{blue}l_1}$-labeled,
 and two, ${\color{red}l_2}$-labeled) and two paths of length $n$, between $v'_2$ and $v'_3$, of each color. 
 Let $K \geq 0$ be the number of independent vertices in $\mathcal{G}$. In $\mathcal{G}'$,
 $\#{\color{blue}l_1} \geq (n+1)(n-K-1) + 2k + n$ and $\#{\color{red}l_2} = 2n + k$. 
 ${\color{red}l_2}=\max\limits_{l\in \mathcal{G}'}(\#l) \Rightarrow K \geq \frac{n^{2}-n-1+k}{n+1} \geq 1$. 
 \begin{figure}[h!]
 \begin{center}
\begin{tikzpicture}[scale=1.6, every node/.style={scale=0.5}]
 \draw[] (1,1) circle (0.5cm);
    \node[] (N1)  at (1,1.4) {$\times$};
    \node[] (N2)  at (1,1.2)   {$\times$};
    \node[] (N3)  at (1,1.05) {$\vdots$};
    \node[] (N4)  at (1,0.8) {$\times$};
    \node[] (N5)  at (1,0.6) {$\times$};
    \node[scale=2] (N6)  at (1.8,1) {$v'_1$};
    \node[scale=2] (N7)  at (2.5,1) {$v'_2$};
    \node[scale=2] (N8)  at (3,1) {$v'_3$};
    \node[scale=2] (N9)  at (2.1,0.7) {$\stackrel{\underbrace{}}{A}$};
    \node[scale=2] (N9)  at (2.7,0.7) {$\stackrel{\underbrace{}}{B}$};
    \node[scale=2] (N10)  at (1,0.3) {$\mathbf{\tilde{\mathcal{G}}}$};
    \node[scale=2] (N11)  at (0.2,1.05) {$\mathbf{\mathcal{G}':}$};
    \path 
    (N1) edge[->, color=red] node {} (N6)
    (N2) edge[->, color=red] node {} (N6)
    (N4) edge[->, color=red] node {} (N6)
    (N5) edge[->, color=red] node {} (N6)
    (N6) edge[->, bend left=30, color=red] node[above,scale=2]{${\color{black}k}$} (N7)
    (N6) edge[->, bend right=30, color=blue] node []{} (N7)
    (N6) edge[->, color=blue] node [] {} (N7)
    (N7) edge[->, bend left=30, color=red] node [above,scale=2] {${\color{black}n}$} (N8)
    (N7) edge[->, color=blue] node {} (N8)
    ;
\end{tikzpicture}
\vspace{-3mm}
\end{center}
\caption{$\mathcal{G}'$ Construction}
\label{gadget}
\end{figure}
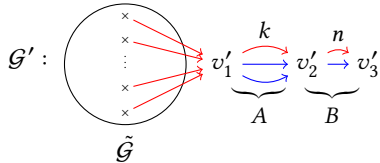
 We show: $\mathcal{G}$ satisfies \textbf{IndSet} $\Leftrightarrow$ $\mathcal{G}'$ satisfies \textbf{MinSummary}.\\     
{\noindent \textbf{$\Rightarrow$} Let $\mathcal{G}$ satisfy \textbf{IndSet}. We can thus choose a set of independent vertices $S \subset V$, $|S| = k$. 
    Let $\mathcal{G}_2$ be the $\mathcal{G}'$-subgraph induced by $S \cup A \cup B$.
    It is \emph{maximal weakly label-connected} on ${\color{red}l_2}$ and contains $2k + n$ edges labeled ${\color{red}l_2}$ and $2k+n$ edges labeled ${\color{blue}l_1}$, i.e.,
    $\lambda(\mathcal{G}_2) = {\color{red}l_2}$. 
    Let $\mathcal{G}_1$ be the $\mathcal{G}'$-subgraph induced by $V \setminus S$. It is \emph{maximal weakly label-connected}
    on ${\color{blue}l_1}$ and contains $(n+1)m$ edges, all labeled ${\color{blue}l_1}$; hence, $\lambda(\mathcal{G}_1) = {\color{blue}l_1}$. 
    $\Phi = \{\mathcal{G}_1, \mathcal{G}_2\}$ is a valid summarization of $\mathcal{G}'$,
    as ${\color{blue}l_1}=\max\limits_{l\in \mathcal{G}_1}(\#l)$ and ${\color{red}l_2}=\max\limits_{l\in \mathcal{G}_2}(\#l)$.
    $\mathcal{G'}$ satisfies \textbf{MinSummary}.} \\
 {\noindent \textbf{$\Leftarrow$} Let $\mathcal{G}'$ satisfy \textbf{MinSummary}. We can thus compute a $\mathcal{G}$-partitioning, $\Phi$,
  that is a \emph{valid summarization}, where $|\Phi| \leq 2$.
  If $\Phi = 2$, then there exist two distinct $\mathcal{G}'$-subgraphs, $\mathcal{G}_1$, $\mathcal{G}_2$, where $\Phi = \{\mathcal{G}_1, \mathcal{G}_2\}$. 
  As $\#{\color{blue}l_1}=(n+1)m + 2k + n \geq 2n + k = \#{\color{red}l_2}$ in $\mathcal{G}'$, one of the subgraphs $\mathcal{G}_1$, $\mathcal{G}_2$, should
  be such that all of its components are \emph{maximal weakly label-connected} on ${\color{blue}l_1}$. Let that subgraph be $\mathcal{G}_1$.
  Hence, $\mathcal{G}_1 \cap \tilde{\mathcal{G}}$ contains all vertices connected by a ${\color{blue}l_1}$-labeled edge.
  Let us denote by $\tilde{V_1}$ the set of vertices in $\mathcal{G}_1 \cap \tilde{\mathcal{G}}$. The set of vertices in $\mathcal{G}_1$
  is thus $\tilde{V_1} \cup A \cup B$. As $\Phi$ has to be a valid summarization, the set of vertices in $\mathcal{G}_2$ is $V_2$, where
  $V_2 = V' \setminus (\tilde{V_1} \cup A \cup B)$. We can thus choose the set of independent vertices of size $K$ in $\mathcal{G}$
  to be $S = V_2$. 
  If $|\Phi|=1$, $\Phi = \{\mathcal{G}'\}$ must be a \emph{valid summarization} of
  $\mathcal{G}'$. As $\mathcal{G}'$ is \emph{maximal weakly connected} on ${\color{red}l_2}$, it must hold that
  ${\color{red}l_2}=\max\limits_{l\in \mathcal{G}'}(\#l)$. Hence, 
  $K \geq 1$ and we can choose the set of independent vertices in $\mathcal{G}$ to be $S = V' \cap V$.} 
  Thus, $\mathcal{G}$ satisfies \textbf{IndSet}.} 
\end{proof}
\end{theorem}

\subsection{Visualizing Graph Summaries}\label{sec:vizg}
Since a node link visualization for the GRASP summary
(as shown in Figure \ref{fig:node-link-treemap} on an example) is not suitable to quickly
understand the label distributions on the 
graph summary, we adopt
a treemap visualization. This enables users to explore the graph summary, inspect the label distributions, and 
better comprehend the graph summary's compression factor and topology. Each
cell in the treemaps of Figure \ref{linked-treemap} (depicting a 5K
LDBC social network) encodes an hypernode (HN) in the graph summary. 
 The treemaps are further enriched
with links, in order to represent the underlying HEs connecting HNs. 
The full list of HEs, each of which is
color-coded, in order to distinguish its label (as shown in the legend), is also reported. 
For a given HE, 
depending on the value of EWeight(e,L), we assign a correspondingly
higher/lower thickness to e (as shown for instance by thickier cyan links in Figure
\ref{linked-treemap}). 
Moreover, the links can be added or removed, by selecting and deselecting them in the legend, thus
allowing to obtain different variants of Figure \ref{linked-treemap}, which correspond to going from 
the right to the left-hand side. 
The linked treemaps are used to let the user better understand the graph summaries and 
to explain the results of the corresponding GRASP technique. In the next section, we discuss their usage during AQP and refer the reader to Section 5, for an empirical evaluation of their utility.  
 
%

\begin{figure*}[t!]
\centering
  \includegraphics[width=\textwidth]{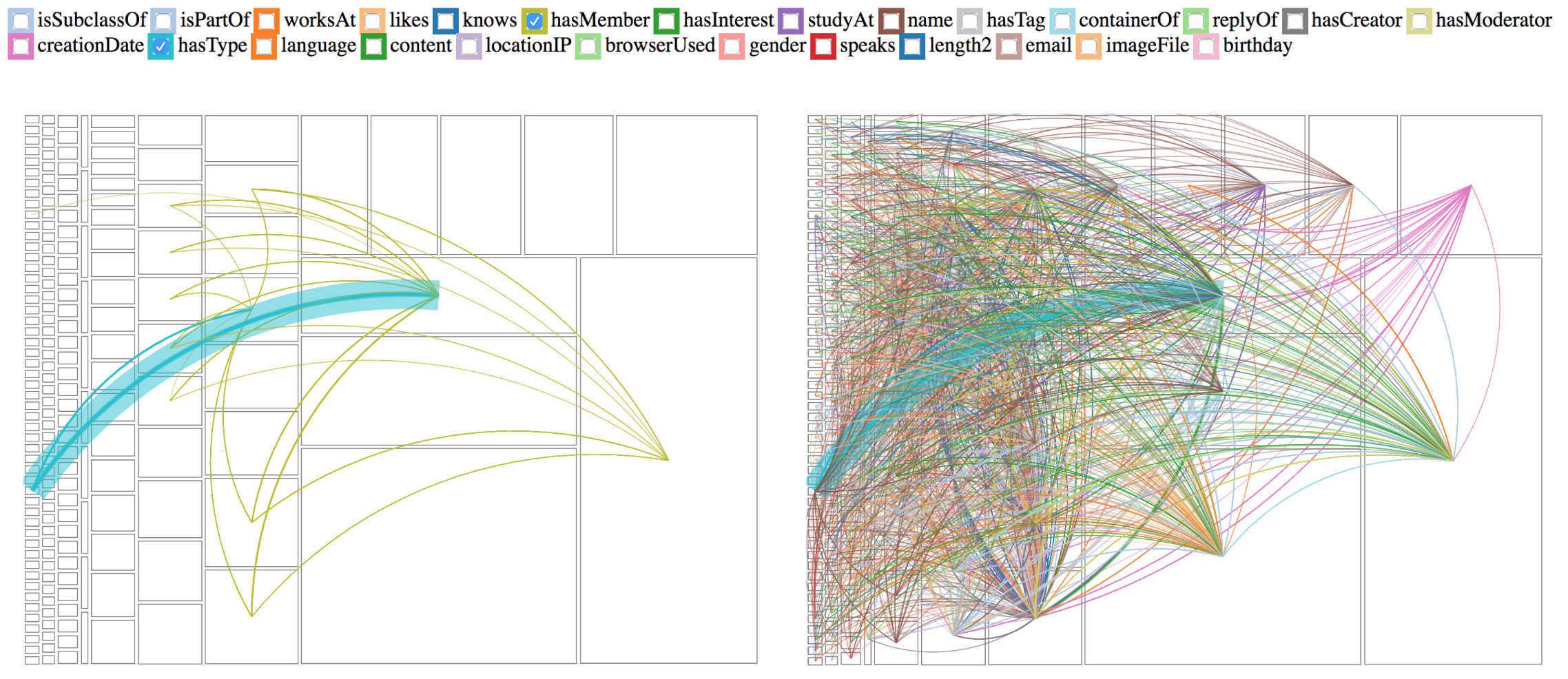}
  \caption{Coarse-grained Linked Treemap for the 5K LDBC Social Network.} 
    \label{linked-treemap}
\end{figure*}


\section{Approximate Query Processing}\label{sec:aqp}
\subsection{Query Translations}\label{sec:qtrans}

For an input graph $\mathcal{G}$ and a counting reachability query $Q$, we aim to approximate the result  
$\llbracket Q \rrbracket_\mathcal{G}$ of evaluating $Q$ over $\mathcal{G}$. Hence, we translate $Q$
into a query $Q^{T}$ to be evaluated over the summarization $\hat{\mathcal{G}}$ of $\mathcal{G}$, such that
$\llbracket Q^{T} \rrbracket_{\hat{\mathcal{G}}} \approx \llbracket Q \rrbracket_{\mathcal{G}}$. The translations 
for each input query type are given in Figure \ref{ref:transl}, using PGQL \cite{RestHKMC16} as a concrete syntax. 
We discuss each query class next. 

{\noindent \textbf{Simple and Optional Label} ($Q_L, Q_O$)} \label{c1}There are two configurations in which the label $l$ can occur in $\hat{\mathcal{G}}$: either inside of
 a HN or on a cross-edge. In the first case, we cumulate the number of $l$-labeled HN inner-edges; in the second,
 we cumulate the $l$-labeled cross-edge weights. To account for the potential absence of $l$, in the optional-label query $Q_2$,
 we additionally estimate the number of nodes in $\mathcal{G'}$, by cumulating the number of nodes in each HN.} 

 {\noindent \textbf{Kleene Plus and Kleene Star} ($Q_P, Q_S$)} \label{c2} To estimate $l^{+}$, we cumulate the counts inside hypernodes containing $l$-labeled inner-edges and, as in (\ref{c1}), the weights on 
 $l$-labeled cross-edges. For the first part, we use the statistics gathered during the \emph{evaluation phase} (Sec. \ref{sec:evalphase}).
 We distinguish three scenarios, depending on whether the $l_{+}$ reachability is due to: 1) inner-edge connectivity -- in which case 
 we use the corresponding property counting the inner $l$-paths; 2) incoming cross-edges -- hence, we cumulate the $l$-labeled in-degrees 
 of HN vertices; or 3) outgoing cross-edges -- in which case we cumulate on the number of outgoing $l$-paths. To handle the $\epsilon$-label 
 in $l^{*}$, we use the same formula as in (\ref{c1}) to additionally estimate the number of nodes in $\hat{\mathcal{G}}$.} 

{\noindent \textbf{Disjunction} ($Q_D$)} \label{c3} As in (\ref{c1}), we treat each configuration, considering both labels. In the first case, we cumulate the number of
 $l_1$ or $l_2$-labeled HN inner-edges; in the second, we cumulate over the cross-edge weights with either label.} 

{\noindent \textbf{Binary Conjunction ($Q_C$)} \label{c4} We consider all cases, depending on whether: 1) the 
label concatenation $\l_1 \cdot l_2$ appears on a path \emph{inside} of a HN, 2) one of the labels $l_1, l_2$ occurs on a 
HN inner-edge and the other, as a cross-edge, or 3) both labels occur on cross-edges.}

 \hspace{-1.5em}
\begin{figure}[h!]\label{ref:qtrans}
\setlength\tabcolsep{1.5pt}
\begin{scriptsize}
\begin{tabular}{|c| @{\hspace{-6\tabcolsep}} l|} \hline
 $Q_L(l)$ & 
  \begin{lstlisting}
    SELECT COUNT(*) MATCH () -[:l]-> ()
  \end{lstlisting} \\ \hline
 $Q^{T}_L(l)$ & 
  \begin{lstlisting}
    SELECT SUM(x.LPERCENT_L * x.EWEIGHT) MATCH (x)    
    SELECT SUM(e.EWEIGHT) MATCH () -[e:l]-> ()
  \end{lstlisting} \\ \hlinewd{1pt}  
 $Q_O(l)$ & 
  \begin{lstlisting}
    SELECT COUNT(*) MATCH () -[:l?]-> ()
  \end{lstlisting} \\ \hline
 $Q^{T}_O$ &   
  \begin{lstlisting}
    SELECT SUM(x.LPERCENT_L * x.EWEIGHT) MATCH (x)
    SELECT SUM(e.EWEIGHT) MATCH () -[e:l]-> ()
    SELECT SUM(x.AVG_SN_VWEIGHT * x.VWEIGHT) MATCH (x) 
  \end{lstlisting} \\ \hlinewd{1pt} 
  
 $Q_P(l)$ & 
  \begin{lstlisting}
    SELECT COUNT(*) MATCH () -/:l+/-> ()
  \end{lstlisting} \\ \hline
 $Q^{T}_P(l)$ & 
  \begin{lstlisting}
    SELECT SUM(x.LREACH_L) MATCH (x) 
    WHERE x.LREACH_L > 0    
    SELECT SUM(e.EWEIGHT) MATCH () -[e:l]-> ()
  \end{lstlisting} \\ \hlinewd{1pt} 
 $Q_S(l)$ &   
  \begin{lstlisting}
    SELECT COUNT(*) MATCH () -/:l*/-> ()
  \end{lstlisting}  \\ \hline
 $Q^{T}_S(l)$ &  
  \begin{lstlisting}
    SELECT SUM(x.LREACH_L) MATCH (x) 
    WHERE x.LREACH_L > 0    
    SELECT SUM(e.EWEIGHT) MATCH () -[e:l]-> ()    
    SELECT SUM(x.AVG_SN_VWEIGHT * x.VWEIGHT) MATCH (x) 
  \end{lstlisting} 
 \\ \hlinewd{1pt}

 $Q_D(l_1,l_2)$ & 
  \begin{lstlisting}
    SELECT COUNT(*) MATCH () -[:l1|l2]-> ()
  \end{lstlisting} \\ \hline
 $Q^{T}_D(l_1, l_2)$ &   
  \begin{lstlisting}
    SELECT SUM(x.LPERCENT_L1 * x.EWEIGHT + 
                x.LPERCENT_L2 * x.EWEIGHT) MATCH (x)
    SELECT SUM(e.EWEIGHT) MATCH () -[e:l1|l2]-> ()
  \end{lstlisting} \\ \hlinewd{1pt}
  $Q_C(l_1,l_2,1,1)$ & 
  \begin{lstlisting}
    SELECT COUNT(*) MATCH () -[:l1]-> () <-[:l2]- ()  
  \end{lstlisting} \\ \hline
  $Q_C(l_1,l_2,1,2)$ & 
  \begin{lstlisting}
    SELECT COUNT(*) MATCH () -[:l1]-> () -[:l2]-> ()
  \end{lstlisting} \\ \hline
  $Q_C(l_1,l_2,2,1)$ & 
  \begin{lstlisting}
    SELECT COUNT(*) MATCH () <-[:l1]- () <-[:l2]- ()  
  \end{lstlisting} \\ \hline
  $Q_C(l_1,l_2,2,2)$ &     
  \begin{lstlisting}
    SELECT COUNT(*) MATCH () <-[:l1]- () -[:l2]-> ()
  \end{lstlisting} \\ \hline
  $Q^{T}_C(l_1,l_2,d_1,d_2)$ &   
  \begin{lstlisting}
    SELECT SUM((x.RLPART_L2_L1_D2_D1 * e.EWEIGHT)/
	    (x.LPERCENT_L1 * x.VWEIGHT))
    MATCH (x) -[e:l2] -> ()
    WHERE x.LPERCENT_L1 > 0   
    SELECT SUM((y.RLPART_L1_L2_D1_D2 * e.EWEIGHT)/
	    (y.LPERCENT_L2 * y.VWEIGHT))
    MATCH () -[e:l1] -> (y)
    WHERE y.LPERCENT_L2 >0   
    SELECT SUM(x.EWEIGHT * 
	    min(x.LPERCENT_L1, x.LPERCENT_L2))
    MATCH (x)    
    SELECT SUM(x.EREACH_L1_L2_D1_D2)  MATCH (x)
  \end{lstlisting} \\ \hlinewd{1pt}  
\end{tabular}
\end{scriptsize}
\caption{Query translations onto the graph summary. }
\label{ref:transl}
\end{figure}

\vspace{-0.2cm}

\begin{example}[Approximating Brand Reach Estimates]
Revisiting Example \ref{ref:ads}, we evaluate the AQP-translation in Table \ref{ref:transl} over the 
GRASP summary $\hat{\mathcal{G}} = (\hat{V}, \hat{E})$ in Fig. \ref{fig:M3}, as follows: 

{\noindent
{\small
 $\llbracket Q_1 \rrbracket_{\hat{\mathcal{G}}} = Q_L^{T}({\color{red}l_5}) = 
 \sum\limits_{\hat{v} \in \hat{V}} EWeight(\hat{v}, {\color{red}l_5}) * LPercent(\hat{v}, {\color{red}l_5})$.\\
 Hence, $\llbracket Q_1 \rrbracket_{\hat{\mathcal{G}}} = EWeight(\text{HN}_2, {\color{red}l_5}) * LPercent(\text{HN}_2, {\color{red}l_5}) = 7$ \\
 $\llbracket Q_2 \rrbracket_{\hat{\mathcal{G}}} = Q_O^{T}({\color{green}l_2}) = Q_L^{T}({\color{green}l_2}) +  
  \sum\limits_{\hat{v} \in \hat{V}} AvgSNVWeight(\hat{v}) * VWeight(\hat{v})$. Hence, $\llbracket Q_2 \rrbracket_{\hat{\mathcal{G}}} =  Q_L^{T}({\color{green}l_2}) = 27$\\
 $\llbracket Q_3 \rrbracket_{\hat{\mathcal{G}}} = Q_P^{T}(l_0) = \sum\limits_{\hat{v} \in \hat{V}} LReach(\hat{v}, l_0) + \sum\limits_{\hat{e} \in \hat{E}} EWeight(\hat{e}, l_0)$\\
 Hence, $\llbracket Q_3 \rrbracket_{\hat{\mathcal{G}}} = \sum\limits_{\hat{v} \in \hat{V}} LReach(\hat{v}, l_0) = 15$. \\
 $\llbracket Q_4 \rrbracket_{\hat{\mathcal{G}}} = Q_S^{T}(l_0) = Q_P^{T}(l_0) + 
  \sum\limits_{\hat{v} \in \hat{V}} AvgSNVWeight(\hat{v}) * VWeight(\hat{v})$.\\ Hence, $\llbracket Q_4 \rrbracket_{\hat{\mathcal{G}}} = 40$. \\
 $\llbracket Q_5 \rrbracket_{\hat{\mathcal{G}}} = Q_D^{T}({\color{orange}l_4}, {\color{blue}l_1}) = 
  Q_L^T ({\color{orange}l_4}) + Q_L^T ({\color{blue}l_1}) = 14$. \\
 $\llbracket Q_6 \rrbracket_{\hat{\mathcal{G}}} = Q_C^{T}({\color{orange}l_4}, {\color{red}l_5}) = 7$.
}
}
\end{example}

Next, we will empirically study the error bounds on various datasets with queries translated according to Figure \ref{ref:transl}.

\subsection{AQP and Visualization}\label{sec:vaqp}
We now discuss the utility of linked treemaps for approximate query
processing. Remember from Section 3.5 that we can filter the links in the
overlay, by using the checkboxes in the legend above Figure
\ref{linked-treemap}. This functionality lets the user understand the
scope of their queries in workload and whether these are restricted to 
one region or to a few hypernodes and/or supernodes in the linked treemaps. 
Narrowing down the graph summary to only cover these queries, allows to further
enhance the efficiency of approximate query evaluation, as shown in Section \ref{sec:exp}.  

\section{Experimental Analysis}\label{sec:exp}
In this section, we present our extensive empirical evaluation recording 
(1) the succinctness of our GRASP summaries and the efficiency of our graph summarization algorithm; 
(2) the suitability of GRASP summaries for approximate evaluation of counting label-constrained reachability queries; 
and (3) the utility of graph visualization in driving the approximate query processing towards certain graph regions, as highlighted in the treemaps. 

\noindent {\bf Setup, Datasets and Implementation.}
Both GRASP and the VAGQP engine are implemented in Java using OpenJDK 1.8. Note that the engine
makes query calls (in PGQL) to Oracle Labs PGX 2.7, as the underlying graph database backend. 
Since the available version of PGX works on homogeneous graphs, rather than
on heterogeneous ones, we padded each node in the graph summary with the
same properties as in the other nodes.  
For the intermediate graph analysis operations (e.g.,
weakly connected components), we used Green-Marl, a domain-specific
language adapted for these tasks, and modified the methods to fit
with the constructions of node properties, as required by our graph
summarization algorithm. Finally, the visualization overlays have been
implemented using D3js.
We base our analysis on four graph datasets (see Figure \ref{fig:gdata}), 
encoding: a Bibliographic network (\emph{bib}), 
the LDBC social network schema \cite{ErlingALCGPPB15} (\emph{social}),
Uniprot knowledge graphs (\emph{uniprot}), and
the WatDiv schema \cite{AlucHOD14} (\emph{shop}).
We obtained these using gMark \cite{BaganBCFLA17}, a synthetic graph instance and query workload generator. 
As gMark tries to construct the instance that best fits the given size parameter and schema constraints, we 
notice variations in the resulting sizes (especially for the very dense graphs \emph{social} and \emph{shop}).
Next, on the same datasets, we generated workloads of varying sizes, for each type in Section
\ref{sec:prelim}, i.e., \emph{single-label, Kleene-star, transitive closure ($+$), unions, and concatenation
queries}. 
Recent studies \cite{BonifatiMT17,Malyshev18} have shown that practical graph pattern queries
formulated by users and in online query endpoints are often small: $56.5\%$
of real-life SPARQL queries consist of a
single edge (RDF triple), whereas $90.8\%$ use 6 edges at most.
Hence, we select small-sized template queries with frequently occuring topologies, such as chains \cite{BonifatiMT17},
and formulate them on our datasets, for workloads of $\sim$ 600 queries.

\begin{figure*}[t!]\centering
\scriptsize
\makebox[0pt]{
\subfloat{
\begin{tabular}{|c|c|c|c|c|c|c|c|c|c|c|c|c|c|c|} \hline
 \multirow{2}{*}{Dataset} & \multirow{2}{*}{$|L_V|$} & \multirow{2}{*}{$|L_E|$} & \multicolumn{2}{|c|}{$\sim 1K$} &  \multicolumn{2}{|c|}{$\sim 5K$}  &  \multicolumn{2}{|c|}{$\sim 25K$} &  \multicolumn{2}{|c|}{$\sim 50K$}  &  \multicolumn{2}{|c|}{$\sim 100K$} &  \multicolumn{2}{|c|}{$\sim 200K$}  \\ 
 \cline{4-15}
 & & & $|V|$ & $|E|$ & $|V|$ & $|E|$ & $|V|$ & $|E|$ & $|V|$ & $|E|$ & $|V|$ & $|E|$ & $|V|$ & $|E|$ \\ \hline
 \multicolumn{1}{|l|}{\emph{bib}}  & $5$  & $4$  & 916 & 1304 & 4565 & 6140 & 22780 & 3159 & 44658 & 60300 & 88879 & 119575 & 179356 & 240052\\\hline
 \multicolumn{1}{|l|}{\emph{social}} & $15$ & $27$ & 897 & 2127 & 4434 & 10896 & 22252 & 55760 & 44390 & 110665 & 88715 & 223376 & 177301 & 450087\\\hlinewd{1.1pt}
 \multicolumn{1}{|l|}{\emph{uniprot}}& $5$  & $7$  & 2170 & 3898 & 6837 & 18899 & 25800 & 97059& 47874 & 192574 & 91600 & 386810 & 177799 & 773082\\\hline
 \multicolumn{1}{|l|}{\emph{shop}}   & $24$ & $82$ & 3136 & 4318 & 6605 & 10811 & 17893 & 34052 & 31181 & 56443 & 57131 & 93780 & 109205 & 168934\\\hline
\end{tabular}}}
\caption{gMark-Generated Dataset Characteristics: number of vertices $|V|$, edges $|E|$, vertex $|L_V|$ and edge labels $|L_E|$.}
\label{fig:gdata}
\end{figure*}


Experiments were executed on a cloud VM with Intel Xeon E312xx (4 cores) 1.80 GHz CPU,
128GB RAM, and Ubuntu 16.04.4 64-bit.
Each data point was obtained by running an experiment $6$ times and
removing the first value from the average computation.  

\noindent {\bf Compression Ratios of GRASP summaries.}
First, we evaluate the effect of using the two heuristics (source-merge and
target-merge) in the construction of the GRASP summaries. We measure the
compression ratio $CR$ obtained on the vertices and edges of the original graph
(by using $(1- |\hat{\mathcal{V}}|/|\mathcal{V}|) * 100$ and $(1-
|\hat{\mathcal{E}|}/|\mathcal{E}|) * 100$, respectively for the CR
vertices and edges), along with the \emph{summary construction time} (SCT). 
Recall that our graph summaries are encoded using the property graph data model and, as such,
they possess node and edge properties. 

Next, we discuss the results for source merge and then compare them with those for target merge. 
In Figures \ref{fig:grasp_exp} (a) and (b), we can observe that the most homogeneous
datasets, (bib) and (uniprot), achieve very high CR (close to $100\%$) and steadily maintain it with varying graph sizes. 
As far as heterogeneity significantly grows for (shop)
and (social), the CR becomes eagerly sensitive to the dataset size, starting with
low values, for smaller graphs, and achieving a
plateau between $85\%$ and $90\%$, for larger ones.
Consequently, our GRASP algorithm enables us to obtain compact summaries for large, highly heterogeneous datasets. 
Notice also that the most heterogeneous datasets, (shop) and (social), although close
to each other, display a symmetric behavior for the CRs of vertices and
edges: the former better compresses vertices, while the latter, edges.
Concerning the SCT runtime in Figure \ref{fig:grasp_exp} (c), all
datasets keep a reasonable performance for larger sizes, including the most heterogeneous one
(shop). The runtime is, in fact, not
affected by the heterogeneity degree, but is rather sensitive, for larger sizes,
to variations in $|E|$ (up to $450K$ and $773K$, for uniprot and social). 

\begin{figure}[h!]
  \begin{center}

        \includegraphics[width=0.95\linewidth]{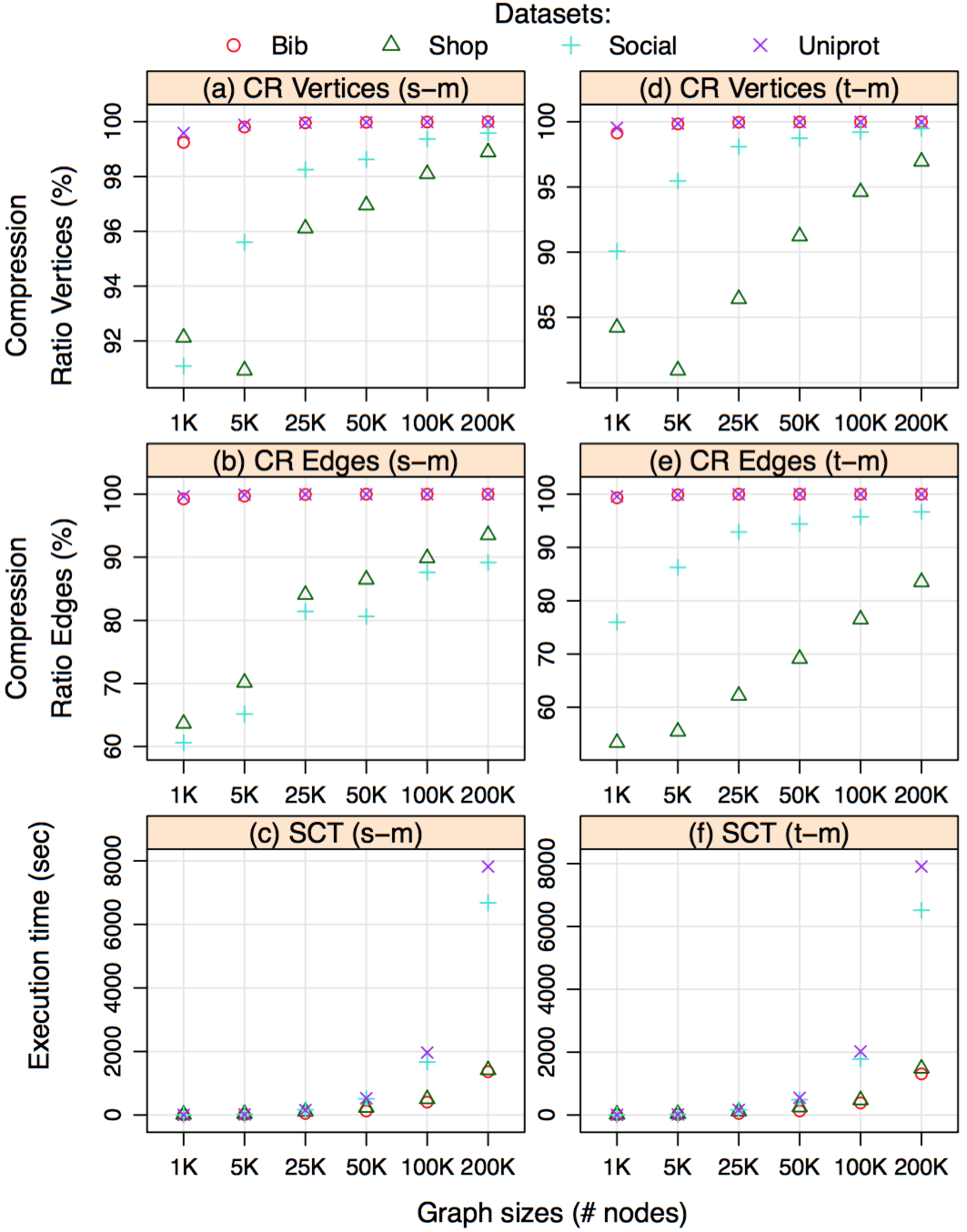} \\
 \caption{Compression Ratios for vertices and edges along with SCT runtime for various sizes of graph datasets for both source-merge (a) to (c) and target-merge (d) to (f).}
    \label{fig:grasp_exp}

  \end{center}
\end{figure}
%
%

We now contrast the source-merge (s-m) and target-merge (t-m) heuristics, the latter being reported in Figure \ref{fig:grasp_exp} (d-e-f). We observe that, while the SCT runtime is quite similar for the two, target-merge achieves better CRs for the social network dataset. 
Overall, the dataset with the worst CR across the two heuristics is shop, which has the lowest CR for smaller sizes. 
This is also due to the high number of labels in the initial shop instances, and hence, to the high number of properties needed for its summary, compared to the other tested datasets: on average, across all considered sizes, 
its summary requires 62.33 properties, against 
17.67 for the social graph one, 
10.0, for bib, and
14.0, for uniprot. 
Nevertheless, even for shop and especially for large sizes, the CRs are fairly high. 
These experiments show that, despite its high complexity, 
GRASP provides high CRs and low SCT runtimes. 

\noindent {\bf AQP Accuracy on GRASP summaries.}
We measured the \emph{accuracy} and \emph{efficiency} of our VAGQP engine by using
the \emph{relative error} and, respectively, the \emph{time gain} measures. The
relative error (per query $Q_i$) is:  $ 1 -
min(Q_i(\mathcal{G}),Q^T_i(\hat{\mathcal{G}}))/$ 
$max(Q_i(\mathcal{G}),Q^T_i(\hat{\mathcal{G}}))$ (in \%), 
where $Q_i(\mathcal{G})$ is the result of the 
counting query $Q_i$ on the original graph (executed with PGX) and $Q^T_i(\hat{\mathcal{G}})$, that of the translated query $Q^T_i$ 
on the GRASP summary (executed with our engine).
The time gain is: $t_{\mathcal{G}} - t_{\hat{\mathcal{G}}}/
max(t_{\mathcal{G}},t_{\hat{\mathcal{G}}})$ (in \%), where the times $t_{\mathcal{G}}$ and $t_{\hat{\mathcal{G}}}$ are the query
evaluation times of query $Q_i$ on the original graph and on the GRASP
summary, respectively.

For the Disjunction, Kleene-plus, Kleene-star, Optional and Single Label query types, we have generated workloads of different sizes,
bound by the number of labels in each dataset. 
For the concatenation workloads, we considered binary conjunctive queries without disjunction, recursion, or optionality. Note that, currently, 
GRASP summaries do not support compositionality. 


\begin{figure*}[t!]
  \begin{center}
    \makebox[0pt]{
    \subfloat[Average Relative Error/Workload]{
        \label{fig:aqp_error_no_concat}
        \includegraphics[width=0.33\linewidth]{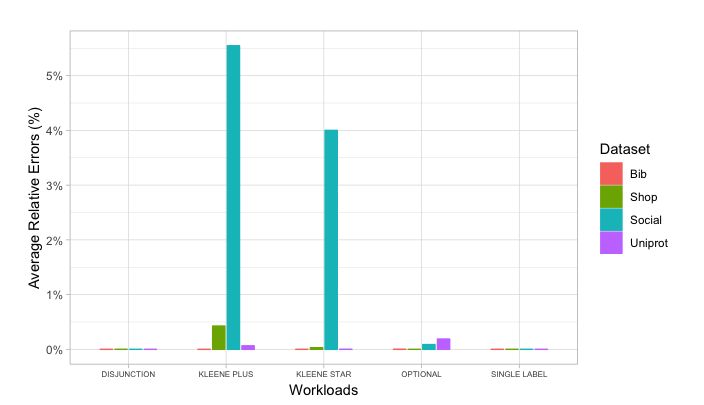}
    }
    \subfloat[Average Time Gain/Workload]{
        \label{fig:aqp_gain_no_concat}
        \includegraphics[width=0.33\linewidth]{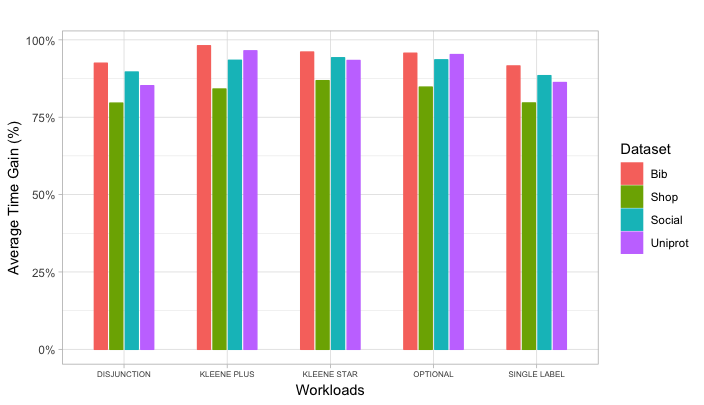}
    }
    \subfloat[Time Gain for Treemap Filtering]{
        \label{fig:runtime_comparison}
        \includegraphics[width=0.33\linewidth]{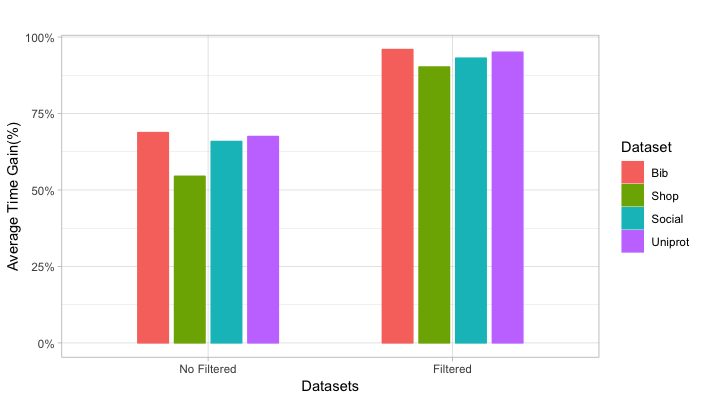}
    }
    }
    \caption{Relative Error (a) and Time Gain (b) per Workload, in each Dataset, for 200K nodes and Single Label, Disjunction, Kleene Star/Plus and Optional Queries, 
    Time Gain (c) for Filtering Coarse-Grained Treemap. }
    \label{fig:aqp_exp_1}

  \end{center}
\end{figure*}

%
%
%
Figure    \ref{fig:aqp_exp_1} (a) and (b) show the relative error and the average time gain
 for Disjunction, Kleene-plus,
Kleene-star, Optional and Single Label workloads. In Figure
  \ref{fig:aqp_exp_1} (a), we can observe that the avg relative error is kept low in all cases and is bounded by $5.5\%$ 
in the case of social dataset's Kleene-plus and Kleene-star workloads. In all the other cases, including Kleene-plus and Kleene-star workloads
for shop dataset, the error is relatively small and near to $0\%$. This experiment confirms the effectiveness of our GRASP summaries for approximate evaluation of graph queries. 
In  Figure  \ref{fig:aqp_exp_1} (b), we studied the efficiency of AQP on GRASP summaries by reporting the time gain (in $\%$) compared with the query evaluation on the original graphs for the four datasets. We can notice a positive time gain (greater or equal to $75\%$) in most cases, but for disjunction. 
Although the AQP relative error is still advantageous for disjunction, disjunctive queries are time-consuming for AQP on GRASP summaries and can take very long for extremely heterogeneous datasets, such as shop (which have the most labels).  
We believe that the performance of these queries will improve with the next PGX version. 

Figures \ref{fig:aqp_exp_2} (a) and (b) show the comparison among the most heterogeneous datasets (shop and social) on workloads of binary conjunctive queries (on a total of 14 queries, 7 per dataset). 
We report the relative error and time gain per query instead of per workload, as before. We can observe in Figure \ref{fig:aqp_exp_2} (a) a relatively small 
error for almost all queries (with an average of 1.6\%), and an upper bound of 8.44\% for query Q5.
Finally, the shop dataset exhibits a relatively small error with an average of 0.14\%. 

Figure \ref{fig:aqp_exp_2} (b) illustrates the fairly high VAGQP time gains for conjunctive queries. 
Prominently, for the social and shop datasets, VAGQP is 81.64\% and 70.95\% faster than query evaluation on the original graph.
This dataset difference is due to the large number of properties and to the high heterogeneity of the latter.


\begin{figure*}[t!]\centering
\scriptsize
\makebox[0pt]{
\subfloat{
\begin{tabular}{|c|c|c|c|c|c|c|c|}\hline
\multirow{2}{*}{ID} & \multirow{2}{*}{Query} & \multicolumn{2}{|c|}{Estimated Answer} &  \multicolumn{2}{|c|}{Relative Error (\%)} &  \multicolumn{2}{|c|}{Runtime (ms)} \\
\cline{3-8}
 & & SumRDF & GRASP &  SumRDF & GRASP & SumRDF & GRASP \\\hline
Q1 & \multicolumn{1}{|l|}{\texttt{SELECT COUNT(*) MATCH (x0)-[:producer]->()<-[:paymentAccepted]-(x1)}} & 75 & 76 & 1.32 & 0.00 & 136.30 & 38.2 \\\hline
Q2 & \multicolumn{1}{|l|}{\texttt{SELECT COUNT(*) MATCH (x0)-[:totalVotes]->()<-[:price]-(x1)}} & 42.4 & 44 & 3.64  & 0.00  & 50.99 & 17 \\\hline
Q3 & \multicolumn{1}{|l|}{\texttt{SELECT COUNT(*) MATCH (x0)-[:jobTitle]->()<-[:keywords]-(x1)}} & 226.7 & 221 & 2.51 & 0.18  & 463.85 & 12.8\\\hline
Q4 & \multicolumn{1}{|l|}{\texttt{SELECT COUNT(*) MATCH (x0)<-[:title]-()-[:performedIn]->(x1)}} & 19.5 & 20 & 2.50 & 0.00 & 831.72 & 8.8 \\\hline
Q5 & \multicolumn{1}{|l|}{\texttt{SELECT COUNT(*) MATCH (x0)-[:artist]->()<-[:employee]-(x1)}} & 143.3 & 133 & 7.19 & 0.37  & 196.77 & 10.6 \\\hline
Q6 & \multicolumn{1}{|l|}{\texttt{SELECT COUNT(*) MATCH (x0)-[:follows]->()<-[:editor]-(x1)}}  & 524 & 528 & 0.38 & 0.48 & 1295.83 & 19\\\hline
\end{tabular}
}}
\caption{SumRDF and GRASP performance: approximate evaluation of binary conjunctive queries on the respective summaries of a shop graph instance
with 31K nodes and 56K edges;
comparing  estimated \emph{cardinality} (number of computed answers), \emph{relative error} w.r.t results computed
on the original graph, and \emph{query runtime}.}
\label{fig:SGExp}
\end{figure*}

\noindent {\bf Baseline for GRASP-based AQP performance.} The closest system to ours is 
SumRDF~\cite{Stefanoni:2018:ECC} (see Section~\ref{sec:relwork}), which, however, operates on a 
\emph{simpler edge-labeled model rather than on property graphs and is tailored for estimating the results of conjunctive queries only}. 
To set a baseline for GRASP-based AQP, we considered the shop dataset in gMark~\cite{BaganBCFLA17}, simulating the WatDiv benchmark~\cite{AlucHOD14}
(also adopted as a benchmark in the SumRDF paper~\cite{Stefanoni:2018:ECC}). From this dataset with 31K nodes and 56K edges, we generated the 
corresponding SumRDF and GRASP summaries.
We noted that GRASP registers a better CR ratio than SumRDF, with \textbf{2737} nodes vs. \textbf{3480} resources and
\textbf{17430} edges vs. \textbf{29621} triples. This comparison is, however, tentative,
as GRASP compresses vertices independently of the edges, while SumRDF returns triples.
We then considered the same type of conjunctive queries analyzed in Figure \ref{fig:aqp_exp_2} and whose syntax is reported in Figure \ref{fig:SGExp}. 
Comparing GRASP vs. SumRDF (see Figure \ref{fig:SGExp}), we recorded an \emph{average relative error} of 
estimation of only \textbf{0.15\%.} vs. \textbf{2.5\%} and an \emph{average query runtime} of only \textbf{27.55 ms} vs. \textbf{427.53 ms}.
The superior accuracy and time performance of GRASP-based AQP show the promise of our approach, motivating our aim to scale 
towards a fully compositional solution.
As SumRDF does not support disjunctions, Kleene-star/plus queries and optional queries, further comparisons were not feasible.

\noindent {\bf Treemap Utility for AQP on GRASP summaries.}
Treemaps allow users to analyze the topology of a GRASP summary, $\hat{\mathcal{G}}$, 
to explore and to navigate its label distribution. 
Next, we highlight a possible use of treemaps to improve query evaluation runtime. Given the translation of a conjunctive query, 
$Q_C^{T}(l_1, l_2, d_1, d_2)$, the coarse-grained view can reveal 
if the labels $l_1$ or $l_2$ are located in a particular region of $\hat{\mathcal{G}}$, as in Figure \ref{linked-treemap} (b), where both labels are concentrated in
the left region of the GRASP summary (each color representing a different label). This information can be useful to filter out the 
graph summary regions that are not inspected by the queries. 
As a test-case for this treemap usage, we have considered queries in 
the first workload (e.g., single label, disjunction, Kleene star, etc.) and showed the runtime improvement 
the coarse-grained treemap refinement. We have observed a non-negligible improvement in all datasets, i.e. 35.73\% for shop, 27.24\% for social,
27.56\% for uniprot and 27.15\% for bib.

\begin{figure}[htp]
  \begin{center}
        \includegraphics[width=0.95\linewidth]{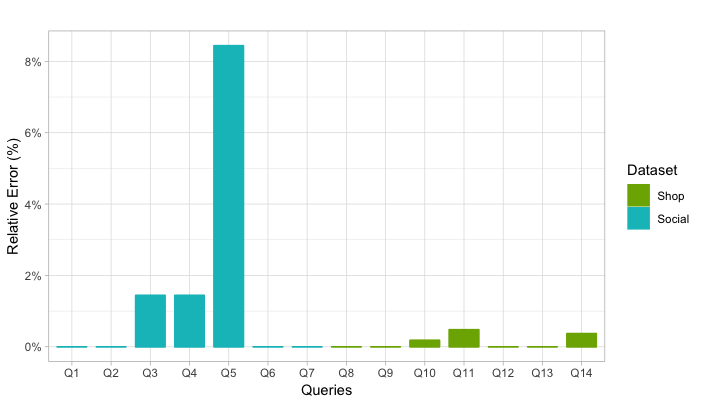} \\
        \includegraphics[width=0.95\linewidth]{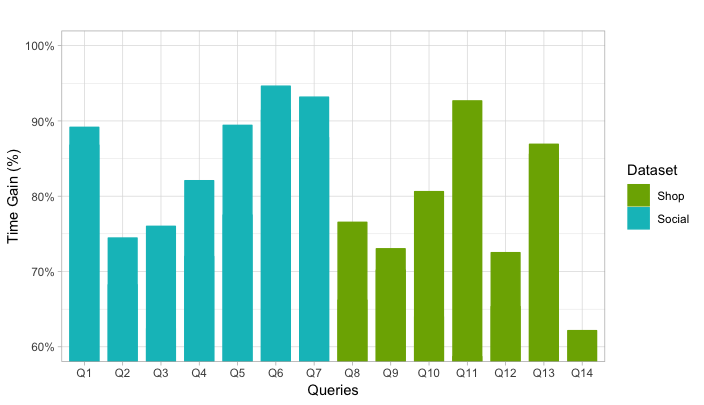}
 \caption{Relative Error (top) and Time Gain (bottom), per Conjunctive Query, in each Dataset, for 50K nodes.}
    \label{fig:aqp_exp_2}

  \end{center}
\end{figure}

\section{Related Work}\label{sec:relwork}
Chaudhuri et al. highlighted in~\cite{ChaudhuriDK17} the two-fold importance of 
AQP in the world of Big Data: (i) to give users agency in deciding
the accuracy vs. efficiency trade-off and (ii) to ensure query-independent accuracy
guarantees. Our work leverages both points, as it targets approximate graph query processing 
and it illustrates how the blend between AQP and graph data
visualization enables advanced graph analytics.

Previous influential work on AQP has focused on relational languages (SQL
and restricted OLAP queries), by embedding samplers directly in the
query language and evaluating them in the query plan~\cite{Acharya2000,
ChaudhuriDK17}.
Relational rows are sampled uniformly-at-random with a
given probability. 
Uniform sampling is widely supported by AQP performing RDBMSs, by 
Big Data systems, e.g., Spark SQL, Oracle, Microsoft's Power BI, SnappyData,
Google's BigQuery, and by online aggregation methods~\cite{Jermaine2008,
Hellerstein1997}. 
Different classes of samplers exist, ranging from the less accurate, uniform ones,
to the distinct and universe ones that work on small
groups and, respectively, on join operators in the query plans. 

Preliminary work on approximate graph analytics in a distributed setting
has recently been pursued in~\cite{IyerPVCASS18}. They rather focus on a graph
sparsification technique and small samples, in order to approximate the
results of specific graph algorithms, such as PageRank and triangle
counting on undirected graphs. 
In contrast, our approach operates in a centralized setting and relies on
query-driven graph summarization for graph
navigational queries with aggregates, while interleaving it with
information visualization techniques devoted to orientate AQP evaluation
towards the regions of the graph pertinent to the queries. 

RDF graph summarization for cardinality  estimation has been tackled in 
recent work~\cite{Stefanoni:2018:ECC}. Their main goal is RDF
query cardinality estimation and  
on a data model that is less expressive than ours
(plain RDF graphs vs. property graphs). 
Hence, the query fragments considered in our and in theirs have limited overlap. As reported in Section~\ref{sec:exp},
our approximate evaluation shows better accuracy and runtime on a set of common queries.   

An algorithm for answering graph reachability queries, using graph simulation based pattern matching, is given
in \cite{FanLMTW11, FanLWW12}, to construct query preserving summaries. This work differs from ours,
as it does not consider property graphs or aggregates.

Aggregation-based graph summarization~\cite{KhanBB17} is at the heart of previous
approaches, the most notable of which is SNAP~\cite{TianHP08}. However, SNAP and
its successor, k-SNAP~\cite{ZhangTP10},are not suitable for AQP
and are mainly devoted to discovery-driven graph summarization of heterogeneous networks 
by relying on a notion of interestingness. These approaches pioneer 
user interaction as a mean to control the resolutions of the graph summary.
However, user intervention here is basically used to drill-down or roll-up to navigate through summaries with
different resolutions.
More recently, preliminary work by Rudolf et al.~\cite{Rudolf14} has
introduced a graph summary
suitable for property graphs based on a set of input summarization rules. In the spirit of SNAP, heir summary is 
conceived for the property graph cube and
supports OLAP operations of the kind roll-up, drill-down and 
slice dice for reducing/expanding the cube dimensions. They further tackle the
problems of unbalanced hierarchies and OLAP anomalies. They do not support
label-constrained reachability queries and VGAQP as we do in this paper.  
In \cite{TangCM16}, a graphical sketch for summarizing graph streams is introduced. 
While it supports dynamic graphs and provides constant maintenance time per update, the 
summary stores no extra information and is not used for AQP.

Visualization has recently been pinpointed as a fertile ground for data
management~\footnote{ACM Sigmod Blog hosts Aditya Parameswaran, 2018}, in the context of the too-many-tuples and
too-many-visualization problems. In our paper, we focus on the graph DBMS case, 
in which these problems are further exacerbated by the fact that interlinking structures 
must be preserved in the evaluation of graph navigational queries.
Optimistic visualization is introduced in the PanGloss
system~\cite{MoritzFD017}, in
order to let the user quickly retrieve AQP results, by only considering a few
relational tuples, and to keep him/her busy until new ones arrive.
Contrary to PanGloss, we do not consider the human-interaction aspect of
VAGQP, which deserves further attention in our future work. 
Graph summaries applied to answer subgraphs, returned by formulating keyword queries on 
large networks, are used in order to enhance the user's understandability of query
results in~\cite{WuYSIY13}. Our query classes are significantly different
from theirs. 

The IDEA system~\cite{GalakatosCZBK17} leverages other visualization
overlays, i.e., heat-maps, where each cell represents the relationship count
between two attributes in a relation. They focus on the rare subpopulation case and on
the construction of suitable indexes.

The AQP++ system~\cite{PengZWP18} blends AQP with aggregate precomputation, such as data cubes, to handle aggregate relational queries. 
This unified approach performs better, in terms of preprocessing, runtime, and quality, than plain AQP.

\section{Conclusion and Perspectives}\label{sec:concl}
This paper presents the GRASP graph summarization technique, suitable for
label-constrained counting reachability queries on property graphs. 
We prove that the problem of deciding whether an optimal graph summary exists, i.e.,
such that the number of label-constrained graph partitions is minimal, is
NP-complete. We then leverage our GRASP summaries and their linked treemap
encoding for AQP purposes. The experiments ran on various datasets show both
fairly high compression ratios, for the summaries, and low relative error and time gain, for the AQP. We also illustrate the advantages of using 
linked treemaps as summary encoding, for both human-explainability and AQP refinement. As future work, we aim to further explore the interplay
between visualization techniques and AQP and to support more graph analytics and aggregation operators. 
As our approach is system-agnostic, being applicable to any framework supporting the property graph model
and the query fragment in Fig.~\ref{fig:graphg}, we would like to implement
it on top of other graph query engines in the near future. Furthermore, we
plan to make open-source
the VAGQP prototype and its visualization plug-in. \\

\balance

\bibliographystyle{abbrv}
\bibliography{paper-bib} 

\end{document}